\documentclass[letterpaper,11pt]{article}
\usepackage[latin1]{inputenc}
\usepackage{amsfonts}
\usepackage{amsmath}
\usepackage{amssymb}
\usepackage{amsthm}
\usepackage{mathptmx}
\usepackage{todonotes}
\usepackage[english]{babel}
\usepackage[T1]{fontenc}
\usepackage[all]{xy}
\usepackage[small]{caption}
\usepackage{hyperref}
\usepackage{verbatim}
\DeclareSymbolFont{AMSb}{U}{msb}{m}{n}
\DeclareMathSymbol{\N}{\mathbin}{AMSb}{"4E}
\DeclareMathSymbol{\Z}{\mathbin}{AMSb}{"5A}
\DeclareMathSymbol{\R}{\mathbin}{AMSb}{"52}
\DeclareMathSymbol{\Q}{\mathbin}{AMSb}{"51}
\DeclareMathSymbol{\I}{\mathbin}{AMSb}{"49}
\DeclareMathSymbol{\C}{\mathbin}{AMSb}{"43}

\newtheorem{theorem}{Theorem}

\newtheorem{lemma}[theorem]{Lemma}

\newtheorem{fact}[theorem]{Fact}
\theoremstyle{definition}
\newtheorem{rmk}{Remark}

\newcommand{\E}{\mathbf{E}}

\newcommand{\eps}{\epsilon}

\DeclareMathOperator{\supp}{\text{Supp}}

\DeclareMathOperator*{\argmax}{arg\,max}

\newcommand{\zmax}{Z_{\max}}
\newcommand{\zq}{\mathbb{Z}_q}

\usepackage{setspace}
\usepackage{graphicx}
\usepackage[margin=1in]{geometry}
\usepackage{xcolor}

\usepackage{srcltx}
\usepackage{thmtools}
\usepackage{thm-restate}

\title{An Entropy Sumset Inequality and 
Polynomially Fast Convergence to Shannon Capacity Over All Alphabets}

\author{Venkatesan Guruswami\thanks{Computer Science Department, Carnegie 
Mellon University, {\tt 
guruswami@cmu.edu}. Part of this work was done while visiting Microsoft 
Research New England. Research supported in part by NSF grants CCF-0963975 and CCF-1422045.} 
\and Ameya Velingker\thanks{Computer Science Department, Carnegie 
Mellon University, {\tt 
avelingk@cs.cmu.edu}. Part of this work was done while visiting Microsoft 
Research New England. Research supported in part by NSF grant CCF-0963975.}}

\begin{document}
\maketitle
\begin{abstract}

We prove a  lower estimate on the increase in entropy when two copies of a conditional random variable $X | Y$, with $X$ supported on $\Z_q=\{0,1,\dots,q-1\}$ for prime $q$, are summed modulo $q$. Specifically, given two i.i.d copies $(X_1,Y_1)$ and $(X_2,Y_2)$ of a pair of random variables $(X,Y)$, with $X$ taking values in $\Z_q$, we show
\[ H(X_1 + X_2 \mid Y_1, Y_2)  - H(X|Y) \ge \alpha(q) \cdot  H(X|Y) (1-H(X|Y)) \]
for some $\alpha(q) > 0$, where $H(\cdot)$ is the normalized (by
factor $\log_2 q$) entropy. In particular, if $X | Y$ is not close to
being fully random or fully deterministic and $H(X| Y) \in
(\gamma,1-\gamma)$, then the entropy of the sum increases by
$\Omega_q(\gamma)$. Our motivation is an effective analysis of the
finite-length behavior of polar codes, for which the linear dependence
on $\gamma$ is quantitatively important.  The assumption of $q$ being
prime is necessary: for $X$ supported uniformly on a proper subgroup
of $\Z_q$ we have $H(X+X)=H(X)$. For $X$ supported on infinite groups
without a finite subgroup (the torsion-free case) and no conditioning,
a sumset inequality for the absolute increase in (unnormalized)
entropy was shown by Tao in~\cite{tao-entropy-sumset}.

\smallskip
We use our sumset inequality to analyze Ar\i kan's construction of polar
codes and prove that for any $q$-ary source $X$, where $q$ is any fixed
prime, and any $\epsilon > 0$, polar codes allow {\em efficient} data
compression of $N$ i.i.d. copies of $X$ into $(H(X)+\epsilon)N$
$q$-ary symbols, {\em as soon as $N$ is polynomially large in
  $1/\epsilon$}. We can get capacity-achieving source codes with
similar guarantees for composite alphabets, by factoring $q$ into
primes and combining different polar codes for each prime in
factorization.

\smallskip
A consequence of our result for noisy channel coding is that for {\em all}
discrete memoryless channels, there are explicit codes enabling
reliable communication within $\epsilon > 0$ of the symmetric Shannon
capacity for a block length and decoding complexity bounded by a
polynomial in $1/\epsilon$. The result was previously shown for the
special case of binary input channels~\cite{GX13,HAU-scaling}, and
this work extends the result to channels over any alphabet.
% (improving upon the previous exponential dependence of the code length on $1/\epsilon$).

\end{abstract}

\newpage
\section{Introduction}
In a remarkable work, Ar{\i}kan~\cite{Arikan09} introduced the
technique of channel polarization, and used it to construct a family
of binary linear codes called polar codes that achieve the symmetric
Shannon capacity of binary-input discrete memoryless channels in the
limit of large block lengths.  Polar codes are based on an elegant
recursive construction and analysis guided by information-theoretic
intuition. Ar{\i}kan's work gave a construction of binary codes, and
this was subsequently extended to general alphabets in
\cite{STA09}. In addition to being an approach to realize Shannon
capacity that is radically different from prior ones, channel
polarization turns out to be a powerful and versatile primitive
applicable in many other important information-theoretic
scenarios. For instance, variants of the polar coding approach give
solutions to the lossless and lossy source coding
problem~\cite{arikan-source,KU-lossy}, capacity of wiretap
channels~\cite{MV11}, the Slepian-Wolf, Wyner-Ziv, and Gelfand-Pinsker
problems~\cite{korada}, coding for broadcast channels~\cite{GAG13},
multiple access channels~\cite{STY13,AT12}, interference
networks~\cite{WS14}, etc. We recommend the well-written survey by
\c{S}a\c{s}o\u{g}lu~\cite{Sasoglu12} for a detailed introduction to polar codes.

The advantage of polar codes over previous capacity-achieving methods (such as 
Forney's concatenated codes that provably achieved capacity) was highlighted in 
a recent work~\cite{GX13} where 
{\em polynomial convergence to capacity} was shown in the  {\em binary} 
case (this was also shown independently in \cite{HAU-scaling}).
Specifically, it was shown that polar codes enable approaching the symmetric 
capacity of binary memoryless channels within an additive gap of $\epsilon$ 
with block length, construction, and encoding/decoding complexity all bounded 
by a polynomially growing function of $1/\epsilon$. 
Polar codes are the first and currently only known construction which provably have this property, thus 
providing a formal complexity-theoretic sense in which they are the first 
constructive capacity-achieving codes.

The main objective of this paper is to extend this result to the
non-binary case, and we manage to do this for {\em all} alphabets in
this work. We stress that the best previously proven complexity bound
for communicating at rates within $\eps$ of capacity of channels with
non-binary inputs was {\em exponential} in $1/\eps$.  The high level
approach to prove the polynomially fast convergence to capacity is
similar to what was done in \cite{GX13}, which is to replace the
appeal to general martingale convergence theorems (which lead to
ineffective bounds) with a more direct analysis of the convergence
rate of a specific martingale of entropies.\footnote{The approach
  taken in \cite{HAU-scaling} to analyze the speed of polarization for
  the binary was different, based on channel Bhattacharyya parameters
  instead of entropies. This approach does not seem as flexible as the
  entropic one to generalize to larger alphabets.} However, the
extension to the non-binary case is far from immediate, and we need to
establish a quantitatively strong ``entropy increase lemma'' (see
details in Section~\ref{sec:entropinc}) over all prime alphabets. 
The corresponding inequality admits an easier proof in the binary case, but 
requires more work for general prime alphabets.  For
alphabets of size $m$ where $m$ is not a prime, we can construct a
capacity-achieving code by combining together polar codes for each
prime dividing $m$.

In the next section, we briefly sketch the high level structure of polar codes, 
and the crucial role played by a certain ``entropy sumset inequality'' in our effective 
analysis. Proving this entropic inequality is the main new component in this work, though additional technical work is needed to glue it together with several other ingredients to yield the overall coding result.

\section{Overview of the Contribution}

In order to illustrate our main contribution, which is an inequality on 
conditional entropies for inputs from prime alphabets, in a simple setting, we 
will focus on the source coding (lossless compression) model in this paper. 
The consequence of our results for channel coding, which is not immediate but follows in a standard manner from compression of sources with side information (see for instance \cite[Sec 2.4]{Sasoglu12}),
 is stated in Theorem~\ref{thm:chcoding}.

Let $\zq = \{0,1,\dots,q-1\}$ denote the additive group of integers modulo $q$. Suppose $X$ is a source (random variable) over 
$\zq$ (with $q$ prime), with entropy $H(X)$ (throughout the paper, by entropy 
we will mean the entropy normalized by a $\lg q$ factor, so that $H(X) \in 
[0,1]$). The source coding problem consists of compressing $N$ i.i.d. copies 
$X_0,X_1,\dots,X_{N-1}$ of $X$ to $\approx H(X) N$ (say $(H(X)+\epsilon) N$) 
symbols from $\zq$. The approach based on channel polarization is to find an 
explicit permutation matrix $A \in \zq^{N \times N}$, such that if 
$(U_0,\dots,U_{N-1})^t = A (X_0,\dots,X_{N-1})^t$, then in the limit of $N \to 
\infty$, for most indices $i$, the conditional entropy $H(U_i | 
U_0,\dots,U_{i-1})$ is either $\approx 0$ or $\approx 1$. Note that the 
conditional entropies at the source $H(X_i | X_0,\dots,X_{i-1})$ are all equal 
to $H(X)$ (as the samples are i.i.d.). However, after the linear transformation 
by $A$, the conditional entropies get {\em polarized} to the boundaries $0$ and 
$1$. By the chain rule and conservation of entropy, the fraction of $i$ for 
which $H(U_i | U_0,\dots,U_{i-1}) \approx 1$ (resp. $\approx 0$) must be 
$\approx H(X)$ (resp. $\approx 1-H(X)$).

The polarization phenomenon is used to compress the $X_i$'s as follows: The 
encoder only outputs $U_i$ for indices $i \in B$ where $B = \{ i \mid H(U_i | 
U_0,\dots,U_{i-1}) > \zeta \}$ for some tiny $\zeta = \zeta(N) \to 0$. The 
decoder (decompression algorithm), called a \emph{successive cancellation 
decoder}, estimates the $U_i$'s in the order $i=0,1,\dots,N-1$. For indices $i 
\in B$ that are output at the encoder, this is trivial, and for other 
positions, the decoder computes the maximum likelihood estimate $\hat{u}_i$ 
of $U_i$, assuming $U_0,\dots,U_{i-1}$ equal $\hat{u}_0,\dots,\hat{u}_{i-1}$, 
respectively. Finally, the decoder estimates the inputs at the source by 
applying the inverse transformation $A^{-1}$ to 
$(\hat{u}_0,\dots,\hat{u}_{N-1})^t$.

The probability of incorrect decompression (over the randomness of the source) 
is upper bounded, via a union bound over indices outside $B$, by $\sum_{i 
\notin B} H(U_i | U_0,\dots,U_{i-1}) \le \zeta N$. Thus, if $\zeta \ll 1/N$, we 
have a reliable lossless compression scheme.
Thus, in order to achieve compression 
rate $H(X)+\epsilon$, we need a polarizing map $A$ for which $H(U_i | U_0,\dots,U_{i-1}) \ll 
1/N$ for at least $1-H(X)-\epsilon$ fraction of indices. This in particular means that $H(U_i | U_0,\dots,U_{i-1}) \approx 0$ or $\approx 1$ for all but a vanishing fraction of indices, which can be compactly expressed as $\E_i \bigl[ H(U_i | U_0,\dots,U_{i-1}) \bigl( 1- H(U_i | U_0,\dots,U_{i-1})\bigr) \bigr] \to 0$ as $n \to \infty$.

Such polarizing maps $A$ are in fact implied by a source coding solution, and 
exist in abundance (a random invertible map works w.h.p.). The big novelty in 
Ar{\i}kan's work is an explicit recursive construction of polarizing maps, 
which further, due to their recursive structure, enable efficient maximum 
likelihood estimation of $U_i$ given knowledge of $U_0,\dots,U_{i-1}$. 

Ar{\i}kan's construction is based on recursive application of the basic
$2 \times 2$ invertible map $K = \left( \begin{smallmatrix} 1 & 1 \\ 0
& 1 \end{smallmatrix} \right)$.\footnote{Subsequent work established
that polarization is a common phenomenon that holds for most choices
of the ``base'' matrix instead of just $K$~\cite{KSU10}.}  While Ar{\i}kan's
original analysis was for the binary case, the same construction based
on the matrix $K$ also works for any prime alphabet~\cite{STA09}. Let
$A_n$ denote the matrix of the polarizing map for $N=2^n$.  In the base case 
$n=1$, the
outputs are $U_0=X_0+X_1$ and $U_1=X_1$. If $X_0,X_1 \sim X$ are
i.i.d., the entropy $H(U_0) = H(X_0+X_1) > H(X)$ (unless
$H(X) \in \{0,1\}$), and by the chain rule $H(U_1 | U_0) < H(X)$,
thereby creating a small separation in the entropies.  Recursively, if
$(V_0,\dots,V_{2^{n-1}-1})$ and $(T_0,\dots,T_{2^{n-1}-1})$ are the
outputs of $A_{n-1}$ on the first half and second half of
$(X_0,\dots,X_{2^n-1})$, respectively, then the output $(U_0,\dots,U_{2^n-1})$
satisfies $U_{2i} = V_i + T_i$ and $U_{2i+1} = T_i$.  If $H_n$ denotes
the random variable equal to $H(U_i | U_0,\dots,U_{i-1})$ for a random
$i \in \{0,1,\dots,2^n-1\}$, then the sequence $\{H_n\}$ forms a
bounded martingale.  The polarization property, namely that $H_n \to \text{Bernoulli}(H(X))$ in the limit of
$n \to \infty$, can be shown by appealing to the martingale convergence
theorem. However, in order to obtain a finite upper bound on
$n(\epsilon)$, the value of $n$ needed for $\E [ H_n(1-H_n) ] \le \eps$ (so that most conditional
entropies to polarize to $< \epsilon$ or $> 1-\epsilon$), we need a
more quantitative analysis.  This was done for the binary case
in \cite{GX13}, by quantifying the increase in entropy $H(V_i + T_i |
V_0, \dots, V_{i-1}, T_0, \dots, T_{i-1}) - H(V_i |
V_0,\dots,V_{i-1})$ at each stage, and proving that the entropies
diverge apart at a sufficient pace for $H_n$ to polarize to $0/1$
exponentially fast in $n$, namely $\E[H_n (1-H_n)] \le \rho^n$ for some absolute constant $\rho < 1$. 

The main technical challenge in this work
is to show an analogous entropy increase lemma for all prime
alphabets. The primality assumption is necessary, because a random variable $X$ uniformly supported on a proper subgroup has $H(X) \notin \{0,1\}$ and yet $H(X+X) = H(X)$.  Formally, we prove:
\begin{restatable}{theorem}{mainthm}
\label{thm:main-intro}
 Let $(X_i,Y_i)$, $i=1,2$ be i.i.d. copies of a correlated random variable 
$(X,Y)$ with $X$ supported on $\zq$ for a prime $q$. Then for some $\alpha(q) > 
0$,
 \begin{eqnarray}
 H(X_1+X_2|Y_1,Y_2) - H(X|Y) \geq \alpha(q) \cdot H(X|Y) (1-H(X|Y)). 
\label{eq:mainthm}
\end{eqnarray}
\end{restatable}
The {\em linear} dependence of the entropy increase on the quantity
$H(X|Y) (1-H(X|Y))$ is crucial to establish a speed of polarization
adequate for polynomial convergence to capacity. A polynomial
dependence is implicit in \cite{sasoglu-q-ary-entropy}, but obtaining
a linear dependence requires lot more care. For the case $q=2$,
Theorem~\ref{thm:main-intro} is relatively easy to establish, as it is
known that the extremal case (with minimal increase) occurs when
$H(X|Y=y) = H(X|Y)$ for all $y$ in the support of $Y$~\cite[Lem
  2.2]{Sasoglu12}. This is based on the so-called ``Mrs. Gerber's
Lemma" for binary input channels~\cite{WZ73,Witsenhausen74}, the
analog of which is not known for the non-binary case~\cite{JA14}.
This allows us to reduce the binary version of \eqref{eq:mainthm} to
an inequality about simple Bernoulli random variables with no
conditioning, and the inequality then follows, as the sum of two
$p$-biased coins is $2p(1-p)$-biased and has higher entropy (unless $p
\in \{0,\tfrac{1}{2},1\}$).  In the $q$-ary case, no such simple
characterization of the extremal cases is known or seems
likely~\cite[Sec 4.1]{Sasoglu12}.  Nevertheless, we prove the
inequality in the $q$-ary setting by first proving two inequalities
for unconditioned random variables, and then handling the conditioning
explicitly based on several cases.

More specifically, the proof technique for Theorem~\ref{thm:main-intro} 
involves using an \emph{averaging} argument to write the left-hand side of 
(\ref{eq:mainthm}) as the expectation, over $y,z\sim Y$, of 
$\Delta_{y,z} = H(X_y+X_z) - \frac{H(X_y)+H(X_z)}{2}$, the entropy increase in 
the sum of random variables $X_y$ and $X_z$ with respect to their average entropy (this increase is called the {\em Ruzsa distance} between the random variables $X_y$ and $X_z$, see \cite{tao-entropy-sumset}).
We then rely on 
 inequalities for \emph{unconditioned} random variables to obtain a lower bound 
for this entropy increase. In general, once needs the entropy increase to be at 
least $c\cdot\min\{H(X_y)(1-H(X_y)), H(X_z)(1-H(X_z))\}$, but for some cases, 
we 
actually need such an entropy increase with respect to a larger \emph{weighted} 
average. Hence, we prove the stronger inequality given by 
Theorem~\ref{thm:wtavg}, which shows such an increase with respect to 
$\frac{2H(X_y) + H(X_z)}{3}$ for $H(X_y)\geq H(X_z)$\footnote{While the weaker 
inequality $H(A+B) \geq \frac{H(A)+H(B)}{2} + c\cdot\min\{H(A)(1-H(A)), 
H(B)(1-H(B))\}$ seems to be insufficient for our approach, it should be noted 
that 
the stronger 
inequality $H(A+B) \geq \max\{H(A),H(B)\} + c\cdot\min\{H(A)(1-H(A)), 
H(B)(1-H(B))\}$ 
is generally not true. Thus, Theorem~\ref{thm:wtavg} provides the right 
middle ground. A limitation of similar spirit for the entropy increase when summing two integer-valued random variables was pointed out in \cite{HAT14}.}. 
Moreover, for some cases of the proof, it suffices to bound 
$\Delta_{y,z}$ from below by $\frac{|H(X_y)-H(X_z)|}{2}$, which is provided by 
Lemma~\ref{lem:maxentropy}, another inequality for unconditional random 
variables.

We note a version of Theorem~\ref{thm:main-intro} (in fact with tight
bounds) for the case of unconditioned random variables $X$ taking
values in a torsion-free group was established by Tao in his work on
entropic analogs of fundamental sumset inequalities in additive
combinatorics~\cite{tao-entropy-sumset} (results of similar flavor for
integer-valued random variables were shown in
\cite{HAT14}). Theorem~\ref{thm:main-intro} is a result in the same
spirit for groups with torsion (and which further handles conditional
entropy). While we do not focus on optimizing the dependence of
$\alpha(q)$ on $q$, pinning down the optimal dependence, especially
for the case without any conditioning, seems like a natural question;
see Remark~\ref{rem:entropy-sumset} for further elaboration.

Given the entropy sumset inequality for conditional random variables, 
we are able to track the  decay of $\sqrt{H_n(1-H_n)}$ and use 
Theorem~\ref{thm:main-intro} to show that for $N = 
\mathrm{poly}(1/\epsilon)$, at most $H(X)+\epsilon$ of the conditional 
entropies $H(U_i | U_0,\dots,U_{i-1})$ exceed $\epsilon$. However, to construct 
a good source code, we need $H(X)+\epsilon$ fraction of the conditional 
entropies to be $\ll 1/N$. This is achieved by augmenting a ``fine'' 
polarization stage that is analyzed using an appropriate Bhattacharyya 
parameter. The details of this step are similar to the binary case, and are included in Appendix~\ref{app:fine-polar}.

The efficient construction of the linear source code (i.e., figuring out 
which entropies polarize very
close to $0$ so that those symbols can be dropped), and the efficient
implementation of the successive cancellation decoder are similar to
the binary case~\cite{GX13} and omitted here. Upon combining these ingredients, we
get the following result on lossless compression with complexity scaling 
polynomially in the gap to capacity:
%footnote{As mentioned earlier, a similar 
%result holds for achieving the symmetric capacity of discrete memoryless 
%channels with $q$-ary inputs.}:
\begin{restatable}{theorem}{lcompthm}
\label{thm:lossless}
Let $X$ be a $q$-ary source for $q$ prime with side information $Y$ (which means $(X,Y)$ is a correlated random variable). Let $0 
< \eps < \frac{1}{2}$.
Then
there exists $N \le (1/\epsilon)^{c(q)}$ for a constant $c(q)
< \infty$ depending only on $q$ and an explicit (constructible in 
$\mathrm{poly}(N)$ time)
matrix $L \in \{0,1\}^{(H(X|Y)+\eps) N \times N}$ such that $\vec{X} = 
(X_0,X_1,\dots,X_{N-1})^t$, formed by taking $N$ i.i.d. copies $(X_0,Y_0), 
(X_1,Y_1), \dots, (X_{N-1},Y_{N-1})$ of $(X,Y)$, can, with high
probability, be recovered from $L\cdot \vec{X}$ and $\vec{Y} = (Y_0, Y_1, 
\dots, Y_{N-1})^t$ in $\mathrm{poly}(N)$ time.
\end{restatable}
Moreover, can obtain Theorem~\ref{thm:lossless} for \emph{arbitrary} 
(not necessarily prime) $q$ with the modification that the map $\Z_q^N \to 
\Z_q^{H(X|Y)+\epsilon)N}$ is no longer linear. This is obtained 
by factoring $q$ into primes and combining polar codes over prime alphabets for 
each prime in the factorization.

\smallskip\noindent {\bf Channel coding.}
Using known methods to construct channel codes from polar source codes for compressing sources with side information (see, for instance, \cite[Sec 2.4]{Sasoglu12} for a nice discussion of this aspect), we obtain the following result for channel coding, enabling reliable communication at rates within an additive gap $\epsilon$ to the {\em symmetric capacity} for discrete memoryless channels over any fixed alphabet, with overall complexity bounded polynomially in $1/\epsilon$. 
Recall that a discrete memoryless channel (DMC) $W$ has a finite input alphabet 
$\mathcal{X}$ and a finite output alphabet $\mathcal{Y}$ with transition 
probabilities $p(y | x)$ for receiving $y \in \mathcal{Y}$ when $x \in 
\mathcal{X}$ is transmitted on the channel. The entropy $H(W)$ of the channel is 
defined to be $H(X|Y)$ where $X$ is uniform in $\mathcal{X}$ and $Y$ is the 
output of $W$ on input $X$; the symmetric capacity of $W$, which is the largest 
rate at which one can reliably communicate on $W$ when the inputs have a uniform 
prior, equals $1-H(W)$. Moreover, it should be noted that if $W$ is a 
\emph{symmetric} DMC, then the symmetric capacity of $W$ is precisely the 
Shannon capacity of $W$.

\begin{restatable}{theorem}{chcodingthm}
 \label{thm:chcoding}
 Let $q\geq 2$, and let $W$ be any  discrete memoryless 
channel capacity with input  alphabet $\Z_q$. 
Then, there exists an $N\leq 
(1/\epsilon)^{c(q)}$ for a constant $c(q) < \infty$ depending only on $q$, 
as well as a deterministic $\mathrm{poly}(N)$ construction of a $q$-ary 
code of block length $N$ and rate at least $1-H(W)-\epsilon$, along with a 
deterministic $N\cdot\mathrm{poly}(\log N)$ time decoding algorithm for the 
code such that the block error probability for communication over $W$ is at 
most $2^{-N^{0.49}}$. Moreover, when $q$ is prime, the constructed codes are linear.
\end{restatable}

The structure of our paper will be as follows. Section~\ref{sec:construction} 
will introduce notation, describe the construction of polar codes, and define 
channels as a tool for analyzing entropy increases for a pair of correlated 
random variables. Section~\ref{sec:entropinc} will then prove our main theorem 
and describe the ``rough'' and ``fine'' polarization results that follow from 
the main theorem and allow us to achieve Theorem~\ref{thm:lossless}. 
The appendix contains basic lemmas about the entropy of random variables 
that will be used in the proof of the main theorem. 
Section~\ref{sec:general-alphabet} shows how polar codes for prime alphabets 
may be combined to obtain a capacity-achieving construction over all alphabets, 
thereby achieving a variant of Theorem~\ref{thm:lossless} over non-prime 
alphabets, as well its channel-coding counterpart, Theorem~\ref{thm:chcoding}.

\section{Construction of Polar Codes}\label{sec:construction}

\noindent {\bf Notation.} We begin by setting some of the notation to be used in the rest of the paper.  We will let $\lg$ denote the base 2 logarithm, 
while $\ln$ will denote the natural logarithm.

For our purposes, unless otherwise stated, $q$ will be a prime integer, and we 
identify $\zq = 
\{0,1,2,\dots,q-1\}$ with the additive group of integers modulo $q$.
 We will generally view $\zq$ as a 
$q$-ary alphabet.

%finite field on $q$ elements (along with the usual $+$ and $\cdot$
%field operations).

Given a $q$-ary random variable $X$ taking values in $\zq$, we let $H(X)$ denote
the \emph{normalized entropy} of $X$:
\[
 H(X) = -\frac{1}{\lg q}\sum_{a\in\zq} \Pr[X=a] \lg(\Pr[X=a]).
\]
In a slight abuse of notation, we also define $H(p)$ for a probability 
distribution $p$. If $p$ is a probability distribution over $\zq$, then we shall
let $H(p) = H(X)$, where $X$ is a random 
variable sampled according to $p$. Also, for nonnegative constants
$c_0, c_1, \dots, c_{q-1}$ summing to 1, we will often write
$H(c_0,\dots,c_{q-1})$ as the entropy of the probability distribution on
$\zq$ that samples $i$ with probability $c_i$. Moreover, for a probability
distribution $p$ over $\zq$, we let $p^{(+j)}$
denote the \emph{$j^\text{th}$ cyclic shift} of $p$, namely, the
probability distribution $p^{(+j)}$ over $\zq$ that satisfies
\[
 p^{(+j)}(m) = p(m-j)
\]
for all $m\in\zq$, where $m-j$ is taken modulo $q$. Note that $H(p) = 
H(p^{(+j)})$
for all $j\in\zq$.

Also, let $\|\cdot\|_1$ denote the $\ell_1$ norm on $\R^q$. In particular, for
two probability distributions $p$ and $p'$, the quantity $\|p-p'\|_1$ will
correspond to twice the total variational distance between $p$ and $p'$.

Finally, given a row vector (tuple) $\vec{v}$, we let $\vec{v}^t$ denote a 
column vector given by the transpose of $\vec{v}$.

\subsection{Encoding Map}
Let us formally define the polarization map that we will use to compress a 
source $X$. Given $n \geq 1$, we define an invertible linear transformation 
$G:\zq^{2^n} \to 
\zq^{2^n}$ by $G=G_n$, where $G_t: \zq^{2^t} \to \zq^{2^t}$, 
$0\leq t\leq n$ is a sequence of invertible linear transformations defined as 
follows: $G_0$ is the identity map on $\zq$, and for any $0\leq k < n$ and 
 $\vec{X} = (X_0, X_1, \dots, X_{2^{k+1}-1})^t$, we recursively define $G_{k+1} 
\vec{X}$ as
\[
 G_{k+1} \vec{X} = \pi_{k+1}(G_k(X_0, \dots, X_{2^k-1}) + G_k(X_{2^k}, 
 \dots, X_{2^{k+1}-1}), G_k(X_{2^k},\dots,X_{2^{k+1}-1})),
\]
where $\pi_{k+1}: \zq^{2^{k+1}}\to \zq^{2^{k+1}}$ is a permutation defined by
\[
 \pi_n(v)_j = \begin{cases} v_i \quad &j=2i\\ v_{i+2^k}\quad &j=2i+1 
\end{cases}.
\]

$G$ also has an explicit matrix form, namely, $G = B_n K^{\otimes n} $, where $K = 
\left(\begin{smallmatrix} 1&1 \\ 0&1 \end{smallmatrix}\right)$, $\otimes$ 
is the Kronecker product, 
and $B_n$ is the $2^n\times 2^n$ bit-reversal permutation matrix for $n$-bit 
strings (see~\cite{arikan-source}).

In our set-up, we have a $q$-ary source $X$, and we let $\vec{X} = (X_0, X_1, 
\dots, X_{2^n - 1})^t$ be a collection of $N=2^n$ i.i.d. samples from $X$. 
Moreover, we encode $\vec{X}$ as $\vec{U} = (U_0, U_1, \dots, U_{2^n-1})^t$, 
given by $\vec{U} = G\cdot\vec{X}$. Note that $G$ only has $0,1$ entries, so each $U_i$ is the sum (modulo $q$) of some subset of the $X_i$'s.

\subsection{Channels}
For purposes of our analysis, we define a \emph{channel} $W = (A; B)$ to be a pair of correlated random 
variables $A,B$; moreover, we define the \emph{channel entropy} of 
$W$ to be $H(W) = H(A|B)$, i.e., the entropy of $A$ conditioned on 
$B$.\footnote{It should be noted $W$ can also be interpreted as a communication 
channel that takes in an input $A$ and outputs $B$ according to some 
conditional probability distribution. This is quite natural in the noisy 
channel coding setting in which one wishes to use a polar code for encoding data
in order to achieve the channel capacity of a symmetric discrete memoryless 
channel. However, since we focus on the problem of source coding (data 
compression) rather than noisy channel coding in this paper, we will simply 
view $W$ as a pair of correlated random variables.}

Given a channel $W$, we can define two channel transformations $-$ and $+$ as 
follows. Suppose we take two i.i.d. copies $(A_0; B_0)$ and $(A_1; B_1)$ of 
$W$. Then, $W^-$ and $W^+$ are defined by
\begin{eqnarray*}
 W^{-} &=& (A_0+A_1; B_0, B_1)\\
 W^{+} &=& (A_1; A_0 + A_1, B_0, B_1).
\end{eqnarray*}
By the chain rule for entropy, we see that
\begin{eqnarray}
 H(W^-) + H(W^+) = 2H(W). \label{eq:splitsum}
\end{eqnarray}
In other words, splitting two copies of $W$ into $W^-$ and $W^+$ preserves the 
total channel entropy. These channels are easily seen to obey 
\[
 H(W^+) \leq H(W) \leq H(W^-).
\]
and the key to our analysis will be quantifying the 
separation in the entropies of the 
two split channels.

The aformentioned channel transformations will help us abstract each step of 
the recursive polarization that occurs in the definition of $G$. Let $W = (X; 
Y)$, where $X$ is a source taking values in $\zq$, and $Y$ can be viewed as 
side information. Then, $H(W) = H(X|Y)$. One special case occurs when $Y=0$, 
which corresponds to an absence of side information.

Note that if start with $W$, then after $n$ 
successive applications of either $W\mapsto W^-$ or $W\mapsto W^+$, we can 
obtain one of $N = 2^n$ possible channels in $\left\{W^s: 
s\in\{+,-\}^n\right\}$. (Here, if $s=s_0 s_1 \cdots s_{n-1}$, with each 
$s_i\in\{+,-\}$, then $W^s$ denotes $(\cdots ((W^{s_0})^{s_1}) {{\cdots )}^ 
{s_{n-2}}})^{s_{n-1}}$). By successive applications of (\ref{eq:splitsum}), we 
know that
\[
 \sum_{s\in\{+,-\}^n} W^s = 2^n H(W) = 2^n H(X|Y).
\]
Moreover, it can be verified (see~\cite{Sasoglu12}) that if $0\leq i < 2^n$ 
has binary representation $\overline{b_{n-1} b_{n-2}\cdots b_0}$ (with $b_0$ 
being the least significant bit of $i$), then
\[
 H(U_i|U_0, U_1, \dots, U_{i-1}, Y_0, Y_1,\dots, Y_{N-1}) = H(W^{s_{n-1} 
s_{n-2} \cdots s_0}),
\]
where
\[
 s_j = \begin{cases} -\quad &\text{if $b_j = 0$}\\ +\quad &\text{if $b_j = 1$} 
\end{cases}.
\]
As shorthand notation, we will define the channel
\[
 W_n^{(i)} =  W^{s_{n-1} s_{n-2} \cdots s_0},
\]
where $s_0, s_1, \dots, s_{n-1}$ are as above.

\cite{STA09} shows that all but a vanishing fraction of the $N$ channels 
$W^s$ will be have channel entropy close to 0 or 1:
\begin{theorem}
 For any $\delta > 0$, we have that
 \[
  \lim_{n\to\infty} \frac{\left| \{s\in\{+,-\}^n: H(W^s)\in (\delta, 
1-\delta)\}\right|}{2^n} = 0.
 \]
\end{theorem}
Hence, one can then argue that as $n$ grows, the fraction of channels 
with channel entropy close to 1 approaches $H(X)$. In other words, for any 
$\delta > 0$, if we let
\begin{eqnarray}
\overline{F}_{n,\delta} = \{i:
H(U_i|U_0,U_1,\dots,U_{i-1},Y_0,Y_1,\dots,Y_{N-1}) > 
1-\delta\}, \label{eq:frozenset}
\end{eqnarray}
then
\[
 \frac{|\overline{F}_{n,\delta}|}{2^n} \to H(X|Y),
\]
as $n\to\infty$. Thus, it can be shown that for any fixed $\epsilon > 0$, there 
exists suitably large $n$ such that $\{U_i\}_{i\in \overline{F}_{n,\delta}}$
gives a 
source coding of $\vec{X} = (X_0, X_1, \dots, X_{N-1})$ (with side 
information $\vec{Y} = (Y_0, Y_1, \dots, Y_{N-1})$ with rate $\leq 
H(X|Y)+\epsilon$.

Our goal is to show that $N=2^n$ can be taken to be just polynomial in 
$1/\epsilon$ in order to obtain a rate $\leq H(X|Y) + \epsilon$.

\subsection{Bhattacharyya Parameter}
\label{sec:bhatt}
In order to analyze a channel $W = (X; Y)$, where $X$ takes values in $\zq$, 
we will define the \emph{$q$-ary source Bhattacharyya parameter}
$\zmax(W)$ of the channel $W$ as
\[
 \zmax(W) = \max_{d\neq 0} Z_d(W),
\]
where
\[
 Z_d(W) = \sum_{x\in \zq} \sum_{y\in\supp(Y)} \sqrt{p(x,y)p(x+d,y)}.
\]
Here, $p(x,y)$ is the probability that $X=x$ and $Y=y$ under the joint
probability distribution $(X,Y)$.

Now, the \emph{maximum likelihood decoder} attempts to decode $x$ given $y$ by
choosing the most likely symbol $\hat{x}$:
\[
 \hat{x} = \argmax_{x'\in \zq} \Pr[X=x'|Y=y].
\]
Let $P_e(W)$ be the probability of an error under maximum likelihood decoding,
i.e., the probability that $\hat{x}\neq x$ (or the defining $\argmax$ for
$\hat{x}$ is not unique) for random $(x,y)\sim (X,Y)$. It is known (see
Proposition 4.7 in~\cite{Sasoglu12}) that $\zmax(W)$ provides an upper bound on
$P_e(W)$:
\begin{lemma}\label{lem:pez}
 If $W$ is a channel with $q$-ary input, then the error probability of the 
maximum-likelihood decoder for a single channel use satisfies
\[
 P_e(W) \leq (q-1) \zmax(W).
\]
\end{lemma}
Next, the following proposition shows how the $\zmax$ operator behaves on the
polarized channels $W^-$ and $W^+$. For a proof, see Theorem 1 
in~\cite{Sasoglu12}.
\begin{lemma} \label{lem:zevolution}
 $\zmax(W^+) \leq \zmax(W)^2$, and $\zmax(W^-) \leq q^3 \zmax(W)$.
\end{lemma}
 Finally, the following lemma shows that $\zmax(W)$ is small whenever $H(W)$ is 
small.
\begin{lemma} \label{lem:zhbound}
 $\zmax(W)^2 \leq (q-1)^2 H(W)$.
\end{lemma}
\noindent The proof follows from Proposition 4.8 of~\cite{Sasoglu12}.

\section{Quantification of Polarization} \label{sec:entropinc}
Our goal is to show ``rough'' polarization of the channel. More precisely, we
wish to show that for some $m = O(\lg(1/\epsilon))$ and constant $K$, we have
\[
 \Pr_i[Z(W_m^{(i)}) \leq 2^{-Km}] \geq 1 - H(W) -\epsilon.
\]
The above polarization result will then be used to show the stronger notion of 
``fine'' polarization, which will establish the polynomial gap to capacity.

The main ingredient in showing polarization is the following theorem, which
quantifies the splitting that occurs with each polarizing step.

\begin{theorem}\label{thm:polarquant}
 For any channel $W = (A;B)$, where $A$ takes values in $\zq$,
we have
\[
 H(W^-) \geq H(W) + \alpha(q) \cdot H(W)(1-H(W)),
\]
where $\alpha(q)$ is a constant depending only on $q$.
\end{theorem}
Theorem~\ref{thm:polarquant} follows as a direct consequence of
Theorem~\ref{thm:main-intro}, which we prove in
Section~\ref{subsec:condent}.  Section~\ref{subsec:uncondent} focuses
on proving Theorem~\ref{thm:wtavg} (tackling the unconditioned case),
which will be used in the proof of Theorem~\ref{thm:main-intro}.

\subsection{Unconditional Entropy Gain}\label{subsec:uncondent}
We first prove some results that provide a lower bound on the normalized 
entropy $H(A+B)$ of a sum of random variables $A,B$ in terms of the 
individual entropies.

\begin{lemma}\label{lem:maxentropy}
 Let $A$ and $B$ be random variables taking values over $\zq$. Then,
\[
 H(A+B) \geq \max\{H(A), H(B)\}.
\]
\end{lemma}
\begin{proof}
 Without loss of generality, assume $H(A)\geq H(B)$. Let $p$ be the underlying
probability distribution for $A$. Let $\lambda_i= \Pr[B=i]$. Then, the
underlying probability distribution of $A+B$ is $\lambda_0 p^{(+0)} + \lambda_1 
p^{(+1)}
+ \cdots + \lambda_{q-1} p^{(+(q-1))}$. The desired result then follows directly
from Lemma~\ref{lem:hconvdiff}.
\end{proof}

The next theorem provides a different lower bound for $H(A+B)$.

\begin{theorem}\label{thm:wtavg}
 Let $A$ and $B$ be random variables taking values over $\zq$ such that 
$H(A)\geq H(B)$. Then,
\[
 H(A+B) \geq \frac{2H(A) + H(B)}{3} + c\cdot\min\{H(A)(1-H(A)), H(B)(1-H(B))\}
\]
for $c = \frac{\gamma_0^3 \lg q}{48 q^5 (q-1)^3 \lg(6/\gamma_0)\lg^2 e}$, where 
$\gamma_0 = \frac{1}{500 (q-1)^4 \lg q}$.
\end{theorem}

\noindent \textbf{Overview of proof.}
The proof of the Theorem~\ref{thm:wtavg} splits into various cases depending 
on where $H(A)$ and $H(B)$ lie. Note that some of these cases overlap. The 
overall idea is as follows. If $H(A)$ and $H(B)$ are both bounded away from 0 
and 1 (Case 2), then the desired inequality follows from the concavity of the 
entropy function, using Lemmas \ref{lem:hconvdiff} and \ref{lem:cyclicdist} 
(note that this uses primality of $q$). Another setting in which the inequality 
can be readily proven is when $H(A)-H(B)$ is bounded away from 0 (which we deal 
with in Cases 4 and 5).

Thus, the remaining cases occur when $H(A)$ and $H(B)$ 
are either both small (Case 1) or both large (Case 3). In the former case, one 
can show that $A$ must have most of its weight on a particular symbol, and 
similarly for $B$ (note that this is why we must choose $\gamma_0 \ll 
\frac{1}{\log q}$; otherwise, $A$ could be, for instance, supported uniformly 
on a set of 
size 2). Then, one can use the fact that a $q$-ary random variable having 
weight 
$1-\epsilon$ has entropy $\Theta(\epsilon \log(1/\epsilon))$ 
(Lemmas~\ref{lem:hwtlowbd} and \ref{lem:hwtupperbd}) in order to prove the 
desired inequality (using Lemma~\ref{lem:lowentropyconvolution}).

For the latter case, we simply show that each of the $q$ symbols of $A$ must 
have weight close to $1/q$, and similarly for $B$. Then, we use the fact that 
such a random variable whose maximum deviation from $1/q$ is $\delta$ has 
entropy $1-\Theta(\delta^2)$ (Lemma~\ref{lem:hdev}) in order to prove the 
desired result (using Lemma~\ref{lem:uniformconv}).
\begin{proof}
 Let $\gamma_0$ be as defined in the theorem statement. Note that we must have 
at least 
one of the following cases:
\begin{enumerate}
 \item $0\leq H(A), H(B) \leq \gamma_0$.
 \item $\frac{\gamma_0}{2} \leq H(A), H(B) \leq 1-\frac{\gamma_0}{2}$.
 \item $1-\gamma_0 \leq H(A), H(B) \leq 1$.
 \item $H(A) > \gamma_0$ and $H(B) < \frac{\gamma_0}{2}$.
 \item $H(A) > 1-\frac{\gamma_0}{2}$ and $H(B) < 1-\gamma_0$.
\end{enumerate}
We treat each case separately.
\ \\

\noindent \underline{Case 1}. Let $\max_{0\leq j < q} \Pr[A=j] = 
1-\epsilon$, where $\epsilon \leq \frac{q-1}{q}$. Note that if $\epsilon \geq 
\frac{1}{e}$, then Fact~\ref{fact:incrfunc} implies that
\begin{eqnarray*}
 H(A) &\geq& -\frac{(1-\epsilon)\lg(1-\epsilon)}{\lg q}\\
 &\geq& \frac{1}{\lg q} 
\cdot \min\left\{-\frac{1}{q}\lg\left(\frac{1}{q}\right), 
-\left(1-\frac{1}{e}\right)\lg\left(1-\frac{1}{e}\right) \right\}\\
&>& \gamma_0,
\end{eqnarray*}
which is a contradiction. Thus, $\epsilon < \frac{1}{e}$.

Now, simply note that if $\epsilon
> \gamma_0\lg q$, then Lemma~\ref{lem:hwtlowbd} and Fact~\ref{fact:incrfunc} 
would
imply that
\[
 H(A) \geq \frac{\epsilon\lg(1/\epsilon)}{\lg q} > \gamma_0,
\]
a contradiction. Hence, we must have $\epsilon \leq \gamma_0\lg q$. Similarly, 
we can write $\max_{0\leq j < q} \Pr[B=j] = 1-\epsilon'$ for some positive
$\epsilon' \leq \gamma_0\lg q$. Then, Lemma~\ref{lem:lowentropyconvolution} 
implies
that
\[
 H(A+B) \geq \frac{2H(A) + H(B)}{3} + \frac{1}{51} H(B)(1-H(B)),
\]
as desired.

\noindent \underline{Case 2}. Let $p$ be the underlying probability
distribution for $A$, and let $\lambda_i = \Pr[B=i]$. Then, the underlying
probability distribution of $A+B$ is $\lambda_0 p^{(+0)} + \lambda_1 p^{(+1)} +
\cdots + \lambda_{q-1} p^{(+(q-1))}$. Let $(i_0, i_1, \dots, i_{q-1})$ be a
permutation of $(0,1,\dots,q-1)$ such that $\lambda_{i_0}\geq\lambda_{i_1}\geq
\cdots \geq\lambda_{i_{q-1}}$.

Since $\lambda_0 + \lambda_1 + \cdots +
\lambda_{q-1} = 1$ and $\max_{0\leq j\leq q-1} \lambda_j = \lambda_{i_0}$, we
have
\begin{eqnarray}
\lambda_{i_0} \geq \frac{1}{q}. \label{lem:lamb0}
\end{eqnarray}
Next, let $\epsilon_0 = \frac{\gamma_0}{6\lg(6/\gamma_0)}$. we claim that
\begin{eqnarray}
\lambda_{i_1} > \frac{\epsilon_0}{q-1}. \label{lem:lamb1}
\end{eqnarray}
Suppose not, for
the sake of contradiction. Then, $\lambda_{i_1}, \lambda_{i_2}, \dots,
\lambda_{i_{q-1}} \leq \frac{\epsilon_0}{q-1}$, which implies that
$\lambda_{i_0} = 1-\sum_{j=1}^{q-1} \lambda_{i_j} \geq 1-\epsilon_0$. Since
$\epsilon_0 \leq
\min\left\{\frac{1}{e},\frac{1}{500},\frac{1}{(q-1)^4}\right\}$,
Lemma~\ref{lem:hwtupperbd} and Fact~\ref{fact:incrfunc} imply that
\begin{eqnarray*}
 H(B) \leq \frac{17\epsilon_0 \lg(1/\epsilon_0)}{12\lg q},
\end{eqnarray*}
which is less than $\frac{\gamma_0}{2}$, resulting in a contradiction.
Thus, (\ref{lem:lamb1}) is true.

Therefore, by Lemma~\ref{lem:hconvdiff} and Lemma~\ref{lem:cyclicdist},
\begin{eqnarray*}
 H(A+B) &=& H(\lambda_0 p^{(+0)} + \lambda_1 p^{(+1)} + \cdots + \lambda_{q-1}
p^{(+(q-1))})\\
&\geq& H(A) + 
\frac{1}{2\lg 
q}\cdot \frac{\lambda_{i_0}\lambda_{i_1}}{\lambda_{i_0}+\lambda_{i_1}}
\|p^{(+i_0)}-p^{(+i_1)}\|_1^2 \\
&\geq& H(A) + \frac{1}{2\lg 
q}\lambda_{i_0}\lambda_{i_1}\|p^{(+i_0)}-p^{(+i_1)}\|^2 \\
&\geq& H(A) + \frac{\lambda_{i_0}\lambda_{i_1}(1-H(p))^2 \lg q}{8q^4(q-1)^2 
\lg^2 e}\\
&=& H(A) + \frac{\lambda_{i_0} \lambda_{i_1} \gamma_0^2 \lg q}{32q^4 (q-1)^2
\lg^2 e}\\
&\geq& \frac{2H(A)+H(B)}{3} + \frac{\epsilon_0 \gamma_0^2 \lg q}{32q^5 (q-1)^3 
\lg^2 e}.
\end{eqnarray*}
Finally, note that $\min\{H(A)(1-H(A)), H(B)(1-H(B))\} \leq \frac{1}{4}$, which 
implies that
\[
\frac{\epsilon_0\gamma_0^2 \lg q}{32q^5(q-1)^3\lg^2 e} \geq 
\frac{\epsilon_0\gamma_0^2 \lg q}{8q^5(q-1)^3 \lg^2 e}
\min\{H(A)(1-H(A)), 
H(B)(1-H(B))\}.
\]
Therefore,
\begin{eqnarray*}
 H(A+B) \geq \frac{2H(A)+H(B)}{3} + c \cdot \min\{H(A)(1-H(A)), H(B)(1-H(B))\},
\end{eqnarray*}
where $c = \frac{\gamma_0^3 \lg q}{48q^5 (q-1)^3 \lg(6/\gamma_0) \lg^2 e}$.

\noindent \underline{Case 3}. Let $\Pr[A=i] = \frac{1}{q} + \delta_i$ for 
$0\leq i\leq q-1$. If $\delta = \max_{0\leq i < q}|\delta_i|$, then by 
Lemma~\ref{lem:hdev}, we have
\begin{eqnarray*}
 1-\gamma_0 \leq H(A) \leq 1 - \frac{q^2(q\ln q -(q-1))}{(q-1)^3\ln q}\delta^2,
\end{eqnarray*}
which implies that
\[
 \delta \leq \sqrt{\frac{\gamma_0 (q-1)^3 \ln q}{q^2(q\ln q - (q-1))}} < 
\frac{1}{2q^2}.
\]
Similarly, if we let $\Pr[B=i] = \frac{1}{q} + \delta'_i$ for all $i$, and 
$\delta' = \max_{0\leq i < q} |\delta'_i|$, then
\[
 \delta' \leq \sqrt{\frac{\gamma_0 (q-1)^3 \ln q}{q^2(q\ln q - (q-1))}} < 
\frac{1}{2q^2}.
\]
Thus, by Lemma~\ref{lem:uniformconv}, we see that
\begin{eqnarray*}
 H(A+B) &\geq& H(A) + \frac{\ln q}{16q^2}\cdot H(A)(1-H(A))\\
 &\geq& \frac{2H(A)+H(B)}{3} + \frac{\ln q}{16q^2}\cdot\min\{H(A)(1-H(A)), 
H(B)(1-H(B))\},
\end{eqnarray*}
as desired.

\noindent \underline{Case 4}. Note that by Lemma~\ref{lem:maxentropy},
\begin{eqnarray*}
H(A+B) - \frac{2H(A)+H(B)}{3} &\geq& H(A) - \frac{2H(A)+H(B)}{3}\\
&=& \frac{H(A)-H(B)}{3}\\
&\geq& \frac{\gamma_0}{6}\\
&\geq& \frac{1}{3} H(B)(1-H(B)).
\end{eqnarray*}

\noindent \underline{Case 5}. As in Case 4, we have that
\begin{eqnarray*}
 H(A+B)-\frac{2H(A)+H(B)}{3} \geq \frac{\gamma_0}{6}.
\end{eqnarray*}
However, this time, the above quantity is bounded from below by
$\frac{1}{3}H(A)(1-H(A))$, which completes this case.
\end{proof}

\subsection{Conditional Entropy Gain}\label{subsec:condent}
Theorem~\ref{thm:polarquant} now follows as a simple consequence of our main 
theorem, which we restate and prove below.
\mainthm*

\begin{rmk}
\label{rem:entropy-sumset}
We have not attempted to optimize the dependence of $\alpha(q)$ on
$q$, and our proof gets $\alpha(q) \ge \frac{1}{q^{O(1)}}$. It is easy
to see that $\alpha(q) \le O(1/\log q)$ even without conditioning
(i.e., when $Y=0$). Understanding what is the true behavior of
$\alpha(q)$ seems like an interesting and basic question about sums of
random variables. For random variables $X$ taking values from a
torsion-free group $G$ and with sufficiently large $H_2(X)$, it is
known that $H_2(X_1+X_2) - H_2(X) \ge \frac{1}{2} - o(1)$ and that
this is best possible~\cite{tao-entropy-sumset}, where $H_2(\cdot)$
denotes the \emph{unnormalized} entropy (in bits). When $G$ is the
group of integers, a lower bound $H_2(X_1+X_2) - H_2(X) \ge g(H_2(X))$
for an increasing function $g(\cdot)$ was shown for all ${\mathbb
  Z}$-valued random variables $X$~\cite{HAT14}.  For groups $G$ with
torsion, we cannot hope for any entropy increase unless $G$ is finite
and isomorphic to $\Z_q$ for $q$ prime (as $G$ cannot have non-trivial
finite subgroups), and we cannot hope for an absolute entropy increase
even for $\Z_q$. So determining the asymptotics of $\alpha(q)$ as a
function of $q$ is the analog of the question studied in
\cite{tao-entropy-sumset} for finite groups.
\end{rmk}

\noindent \textbf{Overview of proof.}
Let $X_y$ denote $X|Y = y$. Then, we use an averaging argument: We reduce the 
desired inequality to providing a lower bound for $\Delta_{y,z} = H(X_y+X_z) - 
\frac{H(X_y)+H(X_z)}{2}$, whose expectation over $y,z\sim Y$ is the left-hand 
side of (\ref{eq:mainthm}). Then, one splits into three cases for small, 
large, and medium values of $H(X|Y)$.

Thus, we reduce the problem to aruguing about unconditional entropies. As a 
first step, one would expect to prove  $\Delta_{y,z} \geq 
\min\{H(X_y)(1-H(X_y)), H(X_z)(1-H(X_z))\}$ and use this in the proof of the 
conditional inequality. However, this inequality turns out to be too weak 
to deal with the case in which $H(X|Y)$ is tiny (case 2). This 
is the reason we require Theorem~\ref{thm:wtavg}, which provides an increase 
for $H(X_y+X_z)$ over a higher \emph{weighted} average instead of the simple 
average of $H(X_y)$ and $H(X_z)$. Additionally, we use the inequality 
$H(X_y+X_z) \geq \max\{H(X_y), H(X_z)\}$ to handle certain cases, and this is 
provided by Lemma~\ref{lem:maxentropy}.

In cases 1 and 3 (for $H(X|Y)$ in the middle and high regimes), the proof idea 
is that either (1) there is a significant mass of $(y,z) \sim Y\times Y$ for 
which $H(X_y)$ and $H(X_z)$ are separated, in which case one can use 
Lemma~\ref{lem:maxentropy} to bound $\E[\Delta_{y,z}]$ from below, or (2) there 
is a significant mass of $y\sim Y$ for which $H(X_y)$ lies away from 0 and 1, 
in which case $H(X_y)(1-H(X_y))$ can be bounded from below, enabling us to use 
Theorem~\ref{thm:wtavg}.

\begin{proof}
 Let $h = H(X|Y)$, and let $c$ be the constants defined in the 
statement of Theorem~\ref{thm:wtavg}. Moreover, let $\gamma_1 = 1/20$ and let
\[
 p = \Pr_y\left[H(X_y) \in \left(\frac{\gamma_1}{2},
1-\frac{\gamma_1}{2}\right)\right].
\]
Also, let $X_y$ denote $X|Y=y$, and let
\[
 \Delta_{y,z} = H(X_y + X_z) - \frac{H(X_y) + H(X_z)}{2}.
\]
Note that Lemma~\ref{lem:maxentropy} implies that $\Delta_{y,z} \geq 0$ for 
all $y,z$. Also, $\E_{y\sim Y, z\sim Y} [\Delta_{y,z}] = H(X_1 + X_2 | Y_1, 
Y_2) - H(X|Y)$. For simplicity, we will often omit the subscript and write
$\E[\Delta_{y,z}]$.

We split into three cases, depending on the value of $h$.

\noindent \underline{Case 1}: $h \in (\gamma_1, 1-\gamma_1)$.

\begin{itemize}
 \item \textbf{Subcase 1}: $p\geq \frac{\gamma_1}{4}$. Note that if $H(X_y) \in
\left(\frac{\gamma_1}{2}, 1-\frac{\gamma_1}{2}\right)$, then $H(X_y) (1-H(X_y))
\geq \frac{\gamma_1}{2}\left(1-\frac{\gamma_1}{2}\right)$. Hence, by
Theorem~\ref{thm:wtavg}, we have
\begin{eqnarray*}
 \E[\Delta_{y,z}] &\geq&
\sum_{\substack{y,z\\ \frac{\gamma_1}{2} < H(X_y), H(X_z) < 
1-\frac{\gamma_1}{2}}}
\Pr[Y=y]\cdot\Pr[Y=z] \\
& & \cdot\left(H(X_y+X_z) -
\frac{2\max\{H(X_y),H(X_z)\}+\min\{H(X_y),H(X_z)\}}{3}\right)\\
&\geq& \sum_{\substack{y,z\\ \frac{\gamma_1}{2} < H(X_y), H(X_z) < 
1-\frac{\gamma_1}{2}}}
\Pr[Y=y]\cdot\Pr[Y=z]\cdot c\\
& & \cdot \min\{H(X_y)(1-H(X_y)), H(X_z)(1-H(X_z))\}\\
&\geq& \frac{c\gamma_1}{2} \left(1-\frac{\gamma_1}{2}\right)
\sum_{\substack{y,z\\ \frac{\gamma_1}{2} < H(X_y), H(X_z) < 
1-\frac{\gamma_1}{2}}} \Pr[Y=y]\cdot\Pr[Y=z]\\
&=& cp^2\cdot \frac{\gamma_1}{2}\left(1-\frac{\gamma_1}{2}\right)\\
&\geq& \frac{c\gamma_1^3}{32}\left(1-\frac{\gamma_1}{2}\right) \\
&\geq& \frac{c\gamma_1^3}{8}\left(1-\frac{\gamma_1}{2}\right) \cdot
h(1-h).
\end{eqnarray*}

\item \textbf{Subcase 2}: $p < \frac{\gamma_1}{4}$. Note that
\begin{eqnarray*}
 \gamma_1 < h &\leq& \Pr_y \left[H(X_y) \leq \frac{\gamma_1}{2}\right] \cdot
\frac{\gamma_1}{2} + \Pr_y \left[H(X_y) > \frac{\gamma_1}{2}\right]\cdot 1\\
&\leq& \frac{\gamma_1}{2} + \Pr_y\left[H(X_y) > \frac{\gamma_1}{2}\right]
\end{eqnarray*}
which implies that
\[
 \Pr_y\left[H(X_y) > \frac{\gamma_1}{2}\right] \geq \frac{\gamma_1}{2}.
\]
Thus,
\begin{eqnarray}
 \Pr_y \left[H(X_y) \geq 1-\frac{\gamma_1}{2}\right] &=& \Pr_y\left[H(X_y)
> \frac{\gamma_1}{2}\right] - \Pr_y \left[\frac{\gamma_1}{2} < H(X_y) < 1 -
\frac{\gamma_1}{2}\right] \nonumber\\
&\geq& \frac{\gamma_1}{2} - p \nonumber\\
&>& \frac{\gamma_1}{4}. \label{eq:highconcbound}
\end{eqnarray}
Also,
\begin{eqnarray*}
 1-\gamma_1 > h \geq \left(1-\frac{\gamma_1}{2}\right) \cdot \Pr_y\left[ H(X_y)
\geq 1-\frac{\gamma_1}{2}\right],
\end{eqnarray*}
which implies that
\[
 \Pr_y\left[ H(X_y)\geq 1-\frac{\gamma_1}{2}\right] <
\frac{1-\gamma_1}{1-\frac{\gamma_1}{2}}.
\]
Hence,
\begin{eqnarray}
 \Pr_y \left[H(X_y) \leq \frac{\gamma_1}{2}\right] &=& 1 - \Pr_y
\left[\frac{\gamma_1}{2} < H(X_y) < 1 - \frac{\gamma_1}{2}\right] -
\Pr_y\left[H(X_y) \geq 1-\frac{\gamma_1}{2}\right] \nonumber\\
&>& 1 - p - \frac{1-\gamma_1}{1-\frac{\gamma_1}{2}} \nonumber\\
&>& 1 - \frac{\gamma_1}{4} - \frac{1-\gamma_1}{1-\frac{\gamma_1}{2}}
\nonumber\\
&\geq& \frac{\gamma_1}{4}. \label{eq:lowconcbound}
\end{eqnarray}
Using Lemma~\ref{lem:maxentropy} along with (\ref{eq:highconcbound}) and
(\ref{eq:lowconcbound}), we now conclude that
\begin{eqnarray*}
 \E [\Delta_{y,z}] &\geq& \sum_{\substack{y,z\\ H(X_y)\geq
1-\frac{\gamma_1}{2}\\ H(X_z) \leq \frac{\gamma_1}{2}}}
\Pr[Y=y]\cdot\Pr[Z=z]\cdot \left| \frac{H(X_y)-H(X_z)}{2}\right|\\
&\geq& \frac{\gamma_1}{4}\cdot\frac{\gamma_1}{4}\cdot\frac{1-\gamma_1}{2}\geq 
\frac{\gamma_1^2 (1-\gamma_1)}{8}\cdot h(1-h),
\end{eqnarray*}
as desired.
\end{itemize}
\noindent \underline{Case 2}: $h\leq\gamma_1$. Then, define $S = \left\{y:
H(X_y) > \frac{4}{5}\right\}$. We split into two subcases.
\begin{itemize}
 \item \textbf{Subcase 1}: $\sum_{y\in S} \Pr[Y=y]\cdot H(X_y) \geq
\frac{2h}{3}$. Then, $\Pr[Y\in S] \geq \frac{2h}{3}$, and so, by
Lemma~\ref{lem:maxentropy}, we have
\begin{eqnarray*}
 \E_{y,z}[H(X_y+X_z)] - h &\geq& \Pr_{\substack{y,z\\ \{y,z\}\cap 
S\neq\emptyset}} \Pr[Y=y]\cdot\Pr[Y=z]\cdot\max\{H(X_y), H(X_z)\} - h\\
&\geq&  \frac{4}{5}(2\cdot\Pr[Y\in S] - \Pr[Y\in S]^2) - h\\
&\geq& \frac{4}{5}\left(2 \cdot \frac{2h}{3} -
\left(\frac{2h}{3}\right)^2\right) - h\\
&=& \frac{1}{15}h\left(1-\frac{16}{3}h\right)\\
&\geq& \frac{1}{15}\left(1-\frac{16\gamma_1}{3}\right)h(1-h).
\end{eqnarray*}

\item \textbf{Subcase 2}: $\sum_{y\in S} \Pr[Y=y]\cdot H(X_y) < \frac{2h}{3}$.
Then,
\begin{eqnarray}
 \sum_{y\not\in S} \Pr[Y=y]\cdot H(X_y) > \frac{h}{3}. \label{eq:wtsum}
\end{eqnarray}
Moreover, observe that $h \geq \frac{4}{5}\cdot\Pr [Y\in S]$, implying that
\begin{eqnarray}
 \Pr [Y\not\in S] \geq 1-\frac{5h}{4}. \label{eq:compset}
\end{eqnarray}

Hence, using Theorem~\ref{thm:wtavg}, (\ref{eq:wtsum}), and 
(\ref{eq:compset}), we find that
\begin{eqnarray*}
 \E[\Delta_{y,z}] &\geq& \sum_{y,z\not\in S} \Pr[Y=y]\cdot\Pr[Y=z]\cdot
\left(\frac{2\max\{H(X_y), H(X_z)\} + \min\{H(X_y),H(X_z)\}}{3}
\right. \\
& & \left. \vphantom{\frac{2\max\{H(X_y), H(X_z)\} +
\min\{H(X_y),H(X_z)\}}{3}} + c\cdot\min\{H(X_y)(1-H(X_y)), H(X_z)(1-H(X_z))\} -
\frac{H(X_y)+H(X_z)}{2}\right)\\
&\geq& \sum_{y,z\not\in S} \Pr[Y=y]\cdot\Pr[Y=z]
\left(\left|\frac{H(X_y)-H(X_z)}{6}\right| +
\frac{c}{5}\cdot\min\{H(X_y),H(X_z)\}\right) \\
&\geq& \sum_{y,z\not\in S} \Pr[Y=y]\cdot\Pr[Y=z]\cdot \left(\frac{H(X_y)}{6} -
\left(\frac{1}{6}-\frac{c}{5}\right)H(X_z)\right)\\
&=& \frac{c}{5} \Pr[Y\not\in S] \cdot \sum_{y\not\in S} \Pr[Y=y]\cdot H(X_y)\\
&>& \frac{c}{5}\left(1-\frac{5h}{4}\right)\cdot\frac{h}{3} \geq c 
\left(\frac{1}{15}-\frac{\gamma_1}{12}\right) h(1-h) ,  \quad \text{as desired.}
\end{eqnarray*}
\end{itemize}
\noindent \underline{Case 3}: $h \geq 1-\gamma_1$. Write $\gamma = 1-h$, and let
\[
 S = \left\{y: H(X_y) > 1-\frac{\gamma}{2}\right\}.
\]
Moreover, let $\overline{S}$ be the complement of $S$. We split into two
subcases.

\begin{enumerate}
 \item \textbf{Subcase 1}: $\Pr_y[y\in S] < \frac{1}{10}$. Then, letting $r =
\Pr_y \left[H(X_y) \leq \frac{1}{10}\right]$, we see that
\begin{eqnarray*}
 h = 1-\gamma &=& \sum_{\substack{y\\ H(X_y)\leq \frac{1}{10}}} \Pr[Y=y] \cdot
H(X_y) 
+ \sum_{\substack{y\\ H(X_y) > \frac{1}{10}}}
\Pr[Y=y]\cdot H(X_y)\\
&\leq& \frac{1}{10}\cdot\Pr_{y}\left[H(X_y)\leq \frac{1}{10}\right] + 
1\cdot\Pr_{y}\left[H(X_y) >\frac{1}{10}\right]\\
&=& \frac{r}{10} + (1-r),
\end{eqnarray*}
which implies that $r\leq\frac{10}{9}\gamma \leq \frac{10}{9} \gamma_1$. Hence, 
letting $T = \left\{y:
\frac{1}{10}\leq H(X_y) \leq 1-\frac{\gamma}{2}\right\}$, we see that
\begin{equation}
 \Pr_y [y\in T]  \geq  1 - \frac{1}{10} - r \geq 
\frac{9}{10}-\frac{10}{9}\gamma_1 \geq \frac{1}{2} \  . \label{eq:yintee}
\end{equation}
Hence, by Theorem~\ref{thm:wtavg} and (\ref{eq:yintee}),
\begin{eqnarray*}
 \E [\Delta_{y,z}] &\geq& \sum_{y,z\in T} 
\Pr[Y=y]\cdot\Pr[Y=z]\cdot\Delta_{y,z}\\
&\geq& \sum_{y,z\in T} \Pr[Y=y]\cdot\Pr[Y=z]\cdot 
\left(c\cdot\min\{H(X_y)(1-H(X_y)), H(X_z)(1-H(X_z))\}\right)\\
&\geq& (\Pr[Y\in T])^2\left(c\cdot
\frac{\gamma}{2}\left(1-\frac{\gamma}{2}\right)\right) \geq 
\frac{c}{8}\gamma\left(1-\frac{\gamma}{2}\right) \geq \frac{c}{8} h(1-h) \ .
\end{eqnarray*}

\item \textbf{Subcase 2}: $\Pr_y[y\in S] \geq \frac{1}{10}$. Then, observe that
by Lemma~\ref{lem:maxentropy},
\begin{eqnarray*}
 \E[\Delta_{y,z}] &\geq& \sum_{\substack{y\in S\\z\in\overline{S}}}
\Pr[Y=y]\cdot\Pr[Y=z]\cdot\frac{H(X_y)-H(X_z)}{2}\\
&=& \frac{\Pr[Y\in\overline{S}] \cdot \sum_{y\in S}\Pr[Y=y]\cdot H(X_y) -
\Pr[Y\in
S]\cdot\sum_{y\in \overline{S}} \Pr[Y=y]\cdot H(X_y)}{2}\\
&=& \frac{\sum_{y\in S}\Pr[Y=y] H(X_y) - (1-\gamma)\Pr[Y\in S]}{2}\\
&\geq& \frac{\left(1-\frac{\gamma}{2}\right)\Pr[Y\in S] - (1-\gamma)\Pr[Y\in
S]}{2}\\
&\geq& \frac{\gamma}{4}\cdot\Pr[Y\in S] \geq \frac{\gamma}{40} \geq \frac{1}{40}h(1-h) \ . \qedhere
\end{eqnarray*}
\end{enumerate}
\end{proof}

\subsection{Rough Polarization}
Now that we have established Theorem~\ref{thm:polarquant}, we are ready to
show rough polarization of the channels $W_n^{(i)}$, $0\leq i < 2^n$, for
large enough $n$. The precise theorem showing rough polarization is as follows.
\begin{theorem} \label{thm:roughpolarize}
 There is a constant $\Lambda < 1$ such that the following holds. For any 
$\Lambda < \rho < 1$, there exists a constant $b_\rho$ such that for all 
channels $W$ with $q$-ary input, all $\epsilon > 0$, and all $n > 
b_\rho \lg(1/\epsilon)$, there exists a set
\[
 \mathcal{W}' \subseteq \{W_n^{(i)}: 0\leq i\leq 2^n - 1\}
\]
such that for all $M\in\mathcal{W}'$, we have $\zmax(M) \leq 2\rho^n$ and 
$\Pr_i[W_n^{(i)}\in\mathcal{W}'] \geq 1-H(W) - \epsilon$.
\end{theorem}
The proof of Theorem~\ref{thm:roughpolarize} follows from the following lemma:
\begin{lemma} \label{lem:sqrtdecay}
 Let $T(W) = H(W)(1-H(W))$ denote the \emph{symmetric entropy} of a channel 
$W$. Then, there exists a constant $\Lambda < 1$ (possibly dependent on 
$q$) such that
\begin{eqnarray}
 \frac{1}{2}\left(\sqrt{T\left(W_{n+1}^{(2j)}\right)} + 
\sqrt{T\left(W_{n+1}^{(2j+1)}\right)}\right) \leq 
\Lambda\sqrt{T\left(W_n^{(j)}\right)} \label{eq:lambda}
\end{eqnarray}
for any $0\leq j < 2^n$.
\end{lemma}
The proof of Lemma~\ref{lem:sqrtdecay} follows from arguments similar to those
in the proof
of Lemma 8 in~\cite{GX13}. For the sake of completeness, we present a complete
proof of Lemma~\ref{lem:sqrtdecay} in Appendix~\ref{sec:polarizeproofs}.

We now show how to prove Theorem~\ref{thm:roughpolarize} from 
Lemma~\ref{lem:sqrtdecay}. Again, the argument follows the one
shown in the proof of Proposition 5 in~\cite{GX13}, 
except that we work with $\zmax$ as opposed to $Z$.
\begin{proof}
 For any $\rho\in (0,1)$, let
 \begin{eqnarray*}
  A_\rho^{l} &=& \left\{i: H(W_n^{(i)}) \leq 
\frac{1-\sqrt{1-4\rho^n}}{2}\right\}\\
  A_\rho^{u} &=& \left\{i: H(W_n^{(i)}) \geq 
\frac{1+\sqrt{1-4\rho^n}}{2}\right\}\\
  A_\rho &=& A_\rho^{l} \cup A_\rho^{u}.
 \end{eqnarray*}
Moreover, note that repeated application of (\ref{eq:lambda}), we have
\begin{eqnarray*}
 \E_i \sqrt{T(W_n^{(i)})} \leq \Lambda^n \sqrt{T(W)} \leq \frac{\Lambda^n}{2}.
\end{eqnarray*}
Thus, by Markov's inequality,
\begin{eqnarray}
 \Pr_i[T(W_n^{(i)}) \geq \alpha] \leq \frac{\Lambda^n}{2\sqrt{\alpha}} 
\label{eq:markov}
\end{eqnarray}

Then, observe that
\begin{eqnarray}
 H(W) &=& \E_i \left[H(W_n^{(i)})\right] \nonumber\\
 &\geq& \Pr[A_\rho^l]\cdot \min_{i\in A_\rho^l} H(W_n^{(i)}) + 
\Pr[A_\rho^u]\cdot\min_{i\in A_\rho^u} H(W_n^{(i)}) + 
\Pr[\overline{A_\rho}]\cdot \min_{i\in\overline{A_\rho}} H(W_n^{(i)}) 
\nonumber\\
&\geq& \Pr[A_\rho^u]\cdot (1-2\rho^n). \label{eq:hwbd}
\end{eqnarray}
Therefore,
\begin{eqnarray}
 \Pr_i\left[H(W_n^{(i)}) \leq 2\rho^n\right] &\geq& \Pr[A_\rho^l] \nonumber\\
 &=& 1 - \Pr[A_\rho^u] - \Pr[\overline{A_\rho}] \nonumber\\
 &\geq& 1 - H(W) - \Pr[A_\rho^u]\cdot 2\rho^n - \Pr[\overline{A_\rho}] 
\label{eq:arho}\\
 &\geq& 1 - H(W) - 2\rho^n - \frac{1}{2}(\Lambda/\sqrt{\rho})^n,
\label{eq:markovapp}
\end{eqnarray}
where (\ref{eq:arho}) follows from (\ref{eq:hwbd}), and (\ref{eq:markovapp}) 
follows from (\ref{eq:markov}). Thus, it is clear that if $\rho > \Lambda^2$, 
then there exists a constant $a_\rho$ such that for $n > 
a_\rho\lg(1/\epsilon)$, we have
\[
 \Pr_i\left[H(W_n^{(i)}) \leq 2\rho^n\right] \geq 1 - H(W) - \epsilon.
\]
To conclude, note that Lemma~\ref{lem:zhbound} implies
\begin{eqnarray*}
 \Pr_i\left[\zmax(W_n^{(i)})\leq 2\rho^n\right] &\geq& \Pr_i\left[H(W_n^{(i)}) 
\leq \frac{4\rho^{2n}}{(q-1)^2} \right]\\
&\geq& \Pr_i\left[H(W_n^{(i)}) \leq 
2\left(\frac{\rho^2}{(q-1)^2}\right)^n\right]\\
&\geq& 1-H(W)-\epsilon
\end{eqnarray*}
for $n>b_\rho \lg(1/\epsilon)$, where $b_\rho = a_{\rho^2/(q-1)^2}$.
\end{proof}

\subsection{Fine Polarization}
Now, we describe the statement of ``fine polarization.'' This is quantified by 
the following theorem.
\begin{restatable}{theorem}{finepolarizethm}
\label{thm:finepolarize}
 For any $0 < \delta < \frac{1}{2}$, there exists a constant $c_\delta$ that 
satisfies the following statement: For any $q$-ary input memoryless 
channel $W$ and $0 < \epsilon < \frac{1}{2}$, if $n_0 > 
c_\delta\lg(1/\epsilon)$, then
\[
 \Pr_i \left[\zmax(W_{n_0}^{(i)}) \leq 2^{-2^{\delta n_0}}\right] \geq 
1-H(W) - \epsilon.
\]
\end{restatable}
The proof follows from arguments similar to those in
\cite{arikan-telatar,GX13}. For the sake of completeness, and because there are some slight differences in the behavior of the $q$-ary Bhattacharyya parameters from Section~\ref{sec:bhatt} compared to the binary case, we present
a proof in Appendix~\ref{app:fine-polar}.

As a corollary, we obtain the following result on lossless compression with
complexity scaling polynomially in the gap to capacity:
\lcompthm*
\begin{proof}
 Let $W = (X; Y)$, and fix $\delta = 0.499$. Also, let $N=2^{n_0}$. Then, by
Theorem~\ref{thm:finepolarize}, for any $n_0 > c_\delta\lg(1/\epsilon)$, we
have that
\[
 \Pr_i \left[\zmax(W_{n_0}^{(i)}) \leq 2^{-2^{\delta n_0}}\right] \geq
1-H(X)-\epsilon.
\]
Moreover, let $N=2^{n_0}$.
Recall the notation in (\ref{eq:frozenset}). Then, letting  $\delta ' =
2^{-2^{\delta n_0}}$, we have that $\Pr_i [i\in\overline{F}_{n_0,\delta'}] \leq
H(X|Y)+\epsilon$ and $Z(W_{n_0}^{(i)}) \geq \delta'$ for all $i\in
\overline{F}_{n_0,\delta'}$. Thus, we can take $L$ to be the linear map
$G_{n_0}$ projected onto the coordinates of $\overline{F}_{n_0,\delta'}$.

By Lemma~\ref{lem:pez} and the union bound, the probability that attempting to
recover $\vec{X}$ from $L\cdot\vec{X}$ and $\vec{Y}$ results in an error is 
given by
\begin{eqnarray}
 \sum_{i\not\in\overline{F}_{n_0,\delta'}} P_e(W_{n_0}^{(i)}) \leq
\sum_{i\not\in\overline{F}_{n_0,\delta'}} (q-1)\zmax(W_{n_0}^{(i)}) \leq
(q-1)N\delta' = (q-1) 2^{n_0 - 2^{\delta n_0}}, \label{eq:unionbd}
\end{eqnarray}
which is $\leq 2^{-N^{0.49}}$ for $N\geq (1/\epsilon)^\mu$ for some
positive constant $\mu$ (possibly depending on $q$). Hence, it suffices to take
$c(q) = 1+\max\{c_\delta, \mu\}$.

Finally, the fact that both the construction of $L$ and the recovery of
$\vec{X}$ from $L\cdot\vec{X}$ and $\vec{Y}$ can be done in $\mathrm{poly}(N)$ 
time follows
in a similar fashion to the binary case (see the binning algorithm
and the successive cancellation decoder in~\cite{GX13} for details). Moreover, 
the entries of $L$ are all in $\{0,1\}$ because of the fact that $L$ can be 
obtained by taking a submatrix of $B_n K^{\otimes n_0}$, where $B_n$ is a 
permutation matrix, and $K = \left(\begin{smallmatrix} 1 & 1\\ 0 & 
1\end{smallmatrix}\right)$ (see~\cite{arikan-source}).
\end{proof}

\section{Extension to Arbitrary Alphabets}\label{sec:arbalph}
\label{sec:general-alphabet}
In the previous sections, we have shown polarization and polynomial gap to 
capacity for polar codes over \emph{prime} alphabets. We now describe how to 
extend this to obtain channel polarization and the explicit 
construction of a polar code with polynomial gap to capacity over 
\emph{arbitrary} alphabets.

The idea is to use the multi-level code construction technique
sketched in~\cite{STA09} (and also recently in \cite{LA14} for
alphabets of size $2^m$). We outline the procedure here. Suppose we
have a channel $W = (X; Y)$, where $X\in \zq$ and
$Y\in\mathcal{Y}$. Moreover, assume that $q = \prod_{i=1}^s q_i$ is
the prime factorization of $q$.

Now, we can write $X = (U^{(1)}, U^{(2)}, \dots, U^{(s)})$, where each 
$U^{(i)}$ 
is a random variable distributed over $[q_i]$. We also define the channels 
$W^{(1)}, W^{(2)}, \dots, W^{(s)}$ as follows: $W^{(j)} = (U^{(j)}; Y, U^{(1)}, 
U^{(2)}, \dots, U^{(j-1)})$. Note that
\begin{eqnarray*}
 H(W) = H(X|Y) &=& H(U^{(1)}, U^{(2)}, \dots, U^{(s)}|Y)\\
 &=& \sum_{j=1}^s H(U^{(j)}|Y, U^{(1)},U^{(2)},\dots, U^{(j-1)}) \\
 &=& \sum_{j=1}^s H(W^{(j)}),
\end{eqnarray*}
which means that $W$ splits into $W^{(1)}, W^{(2)}, \dots, W^{(s)}$. Since 
each $W^{(j)}$ is a channel whose input is over a prime alphabet, one can 
polarize each $W^{(j)}$ separately using the procedure of the previous 
sections. More precisely, the encoding procedure is as follows. For $N$ large 
enough (as specified by Theorem~\ref{thm:lossless}), we take $N$ copies 
$(X_0; Y_0), (X_1; Y_1), \dots, (X_{N-1}; Y_{N-1})$ of $W$, 
where $X_i = (U_i^{(1)}, U_i^{(2)}, \dots U_i^{(s)})$. Then, 
sequentially for $j=1,2,\dots,s$, we encode $U_0^{(j)}, U_1^{(j)}, \dots, 
U_{N-1}^{(j)}$ using $\left\{\left(Y_i, U_i^{(1)}, U_i^{(2)}, \dots, 
U_i^{(j-1)}\right)\right\}_{i=0,1,\dots,N-1}$ as side information (which can 
be done using the procedure of the previous sections, since $U_j$ is a source 
over a prime alphabet).

For decoding, one can simply use $s$ stages of the successive cancellation 
decoder. In the $j^\text{th}$ stage, one uses the successive cancellation 
decoder for $W^{(j)}$ in order to decode $U_0^{(j)}, U_1^{(j)}, \dots, 
U_{N-1}^{(j)}$, assuming that $\left\{U_i^{(k)}\right\}_{k < j}$ has been 
recovered correctly from the previous stages of successive canellation decoding. 
Note that the 
error probability in decoding $X_0, X_1, \dots, X_{N-1}$ can be obtained by 
taking a union bound over the error probabilities for each of the $s$ stages of 
successive cancellation decoding. Since each individual error probability is 
exponentially small (see (\ref{eq:unionbd})), it follows that the overall error 
probability is also negligible.

As a consequence, we obtain Theorem~\ref{thm:lossless} for non-prime $q$, with 
the additional modification that the map $\Z_q^N \to 
\Z_q^{H(X|Y)+\epsilon)N}$ is not linear. Moreover, using the 
translation from source coding to noisy channel coding (see~\cite[Sec 
2.4]{Sasoglu12}), we obtain the following result for channel coding.

\chcodingthm*

\begin{rmk}
 If $q$ is prime, then the $q$-ary code of Theorem~\ref{thm:chcoding} is, in 
fact, linear.
\end{rmk}

\section*{Acknowledgments}

We thank Emmanuel Abbe, Eren \c{S}a\c{s}o\u{g}lu, and Patrick Xia for useful discussions about various aspects of polar codes.

\bibliographystyle{alpha}
\bibliography{polar}
\appendix
\section{Basic Entropic Lemmas and Proof}

For a random variable $X$ taking values in $\zq$, let $H(X)$ denote
the entropy of $X$, normalized to the interval $[0,1]$. More formally, if $p$
is the probability mass function of $X$, then
\[
 H(X) = \frac{1}{\lg q}\sum_{i=1}^q p(i) \lg (p(i))
\]

Moreover, note for the lemmas and theorems in this section, $q \geq 2$ is an 
integer. We do not make any primality assumption about $q$ anywhere in this 
section with the exception of Lemma~\ref{lem:cyclicdist}.

\begin{lemma}\label{lem:hcnvx}
If $X$ and $Y$ are random variables taking values in $\zq$,
then
\[
 H(\alpha X + (1-\alpha)Y) \geq \alpha H(X) + (1-\alpha) H(Y) + \frac{1}{2\lg
q}\alpha(1-\alpha) \|X-Y\|_1^2.
\]
\end{lemma}
\begin{proof}
 This follows from the fact that $-H$ is a $\frac{1}{\lg q}$-strongly convex
function with respect to the $\ell_1$ norm on
\[
\{x=(x_1, x_2, \dots, x_q)\in\R^q: x_1,x_2,\dots,x_q \geq 0, \|x\|_1\leq 1 \}
\]
\end{proof}
\noindent (see Example 2.5 in \cite{Shalev-Shwartz12} for details).
\begin{lemma}\label{lem:hconvdiff}
 Let $p$ be a distribution over $\zq$. Then, if
$\lambda_0,\lambda_1,\dots,\lambda_{q-1}$ are nonnegative numbers adding up to
1, we have
\begin{eqnarray*}
 H(\lambda_0 p^{(+0)} + \lambda_1 p^{(+1)} + \cdots + \lambda_{q-1} 
p^{(+(q-1))}) \geq H(p) +
\frac{1}{2\lg q}\cdot \frac{\lambda_i \lambda_j}{\lambda_i + \lambda_j} 
\|p^{(+i)} -
p^{(+j)}\|_1^2,
\end{eqnarray*}
for any $i\neq j$ such that $\lambda_i + \lambda_j > 0$.
\end{lemma}
\begin{proof}
Note that if $\lambda_i + \lambda_j > 0$, then we have that by
Lemma~\ref{lem:hcnvx},
\begin{eqnarray*}
 H\left(\sum_{k=0}^{q-1} \lambda_k p^{(+k)}\right) &=&
H\left(\sum_{k\neq i,j} \lambda_k p^{(+k)} + (\lambda_i+\lambda_j)
\left(\frac{\lambda_i}{\lambda_i+\lambda_j}p^{(+i)} +
\frac{\lambda_j}{\lambda_i+\lambda_j}p^{(+j)}\right) \right)\\
&\geq& \sum_{k\neq i.j} \lambda_k H(p^{(+k)}) +
(\lambda_i+\lambda_j)H\left(\frac{\lambda_i}{\lambda_i+\lambda_j}p^{(+i)} +
\frac{\lambda_j}{\lambda_i+\lambda_j}p^{(+j)}\right)\\
&=& (1-\lambda_i-\lambda_j)H(p) +
(\lambda_i+\lambda_j)\left(\frac{\lambda_i}{\lambda_i + \lambda_j}H(p^{(+i)}) +
\frac{\lambda_j}{\lambda_i+\lambda_j}H(p^{(+j)})\right)\\
& & + (\lambda_i+\lambda_j)\cdot\frac{1}{2\lg
q}\cdot\frac{\lambda_i}{\lambda_i+\lambda_j}\cdot\frac{\lambda_j}{
\lambda_i+\lambda_j}\cdot \|p^{(+i)}-p^{(+j)}\|_1^2\\
&=& H(p) + \frac{1}{2\lg
q}\cdot\frac{\lambda_i\lambda_j}{\lambda_i+\lambda_j}\cdot\|p^{(+i)}-p^{
(+j)}\|_1^2 ,
\end{eqnarray*}
as desired.
\end{proof}

\begin{lemma} \label{lem:cyclicdist}
 Let $p$ be a distribution over $\zq$, where $q$ is prime. Then,
\[
 \|p^{(+i)}-p^{(+j)}\|_1 \geq \frac{(1-H(p))\lg q}{2q^2 (q-1)\lg e}.
\]
\end{lemma}
\noindent See Lemma 4.5 of \cite{Sasoglu12} for a proof of the above lemma.

\begin{lemma}\label{lem:tail}
 There exists an $\epsilon_1 > 0$ such that for any $0 < \epsilon \leq
\epsilon_1$,
we 
have
\[
 -(1-\epsilon)\lg(1-\epsilon) \leq -\frac{1}{6} \epsilon\lg\epsilon.
\]
\end{lemma}
\begin{proof}
 By L'H\^{o}pital's rule,
 \[
  \lim_{\epsilon\to 0^+}
\frac{(1-\epsilon)\lg(1-\epsilon)}{\epsilon\lg\epsilon} = \lim_{\epsilon\to
0^+} 
\frac{(1-\epsilon)\ln(1-\epsilon)}{\epsilon\ln\epsilon} = \lim_{\epsilon\to 
0^+} \frac{-1 - \ln(1-\epsilon)}{1 + \ln\epsilon} = 0,
 \]
This implies the claim.
\end{proof}
\begin{rmk}
One can, for instance, take $\epsilon_1 = \frac{1}{500}$ in the above lemma.
\end{rmk}
The following claim states that for sufficiently small $\epsilon$, the quantity
$\epsilon \lg\left(\frac{q-1}{\epsilon}\right)$ is close to
$-\epsilon\lg\epsilon$. We omit the proof, which is rather straightforward.
\begin{fact}\label{fact:qdom}
 Let $\epsilon_2 = \frac{1}{(q-1)^4}$. Then, for any $0 < \epsilon \leq
\epsilon_2$,
we have
\[
 \epsilon\lg\left(\frac{q-1}{\epsilon}\right) \leq \frac{5}{4}
\epsilon\lg(1/\epsilon).
\]
\end{fact}
We present one final fact.
\begin{fact}\label{fact:incrfunc}
 The function $f(x) = x\lg(1/x)$ is increasing on the interval $(0, 1/e)$ and 
decreasing on the interval $(1/e, 1)$.
\end{fact}
\begin{proof}
 The statement is a simple consequence of the fact that $f'(x) = \frac{1}{\ln
2} (-1+\ln(1/x))$ is positive on the interval $(0,1/e)$ and negative on the 
interval $(1/e, 1)$.
\end{proof}

\subsection{Low Entropy Variables}
Now, we prove lemmas that provide bounds on the entropy of a
probability distribution
that samples one symbol in $\zq$ with high probability, i.e., a
distribution that 
has low entropy.
\begin{lemma}\label{lem:hwtlowbd}
Suppose $0 < \epsilon < 1$. If $p$ is a distribution on $\zq$ with
mass $1-\epsilon$ on one symbol, then
\[
 H(p) \geq \frac{\epsilon\lg(1/\epsilon)}{\lg q}.
\]
\end{lemma}
\begin{proof}
Recall that the normalized entropy function $H$ is concave. Therefore,
\[
 H(p) \geq H(1-\epsilon, \epsilon, \underbrace{0, 0, \dots, 0}_{q-2}).
\]
Note that
\[
 H(1-\epsilon, \epsilon, \underbrace{0, 0, \dots, 0}_{q-2}) =
\frac{1}{\lg q}(-(1-\epsilon)\lg(1-\epsilon) - \epsilon\lg\epsilon) \geq
 \frac{-\epsilon\lg\epsilon}{\lg q},
\]
which establishes the claim.
\end{proof}

\begin{lemma}\label{lem:hwtupperbd}
 Suppose $0 < \epsilon \leq \min\{\epsilon_1, \epsilon_2\}$, where $\epsilon_1
= 
\frac{1}{500}$ and $\epsilon_2 = \frac{1}{(q-1)^{4}}$. If $p$ is a
distribution on $\zq$ with mass $1-\epsilon$ on one symbol, then
\[
 H(p) \leq \frac{17\epsilon\lg(1/\epsilon)}{12\lg q}.
\]
\end{lemma}
\begin{proof}
By concavity of the normalized entropy function $H$, we have that
\[
 H(p) \leq H\left(1-\epsilon, \frac{\epsilon}{q-1}, \frac{\epsilon}{q-1}, \dots,
\frac{\epsilon}{q-1}\right).
\]
Moreover,
\begin{eqnarray*}
 H\left(1-\epsilon, \frac{\epsilon}{q-1}, \frac{\epsilon}{q-1}, \dots,
\frac{\epsilon}{q-1}\right) &=& \frac{1}{\lg
q}\left(-(1-\epsilon)\lg(1-\epsilon) +
(q-1)\cdot\left(\frac{\epsilon}{q-1}\lg\frac{q-1}{\epsilon}\right)\right)\\
&=& \frac{-(1-\epsilon)\lg(1-\epsilon)}{\lg q} +
\frac{\epsilon\lg\left(\frac{q-1}{\epsilon}\right)}{\lg q}.
\end{eqnarray*}
By Lemma~\ref{lem:tail} (and the remark following it) and Fact~\ref{fact:qdom},
the above quantity is bounded
from above by
\begin{eqnarray*}
 \frac{\frac{1}{6}\epsilon\lg(1/\epsilon)}{\lg q} +
\frac{\frac{5}{4}\epsilon\lg(1/\epsilon)}{\lg q} =
\frac{17\epsilon\lg(1/\epsilon)}{12\lg q},
\end{eqnarray*}
as desired.
\end{proof}

\begin{rmk}\label{rmk:lowent}
 Lemmas~\ref{lem:hwtlowbd} and \ref{lem:hwtupperbd} show that for sufficiently 
small $\epsilon$, a random variable $X$ over $\zq$ having weight $1-\epsilon$ 
on a particular symbol in $\zq$ has entropy $\Theta(\epsilon\lg(1/\epsilon)/\lg 
q)$. This allows us to prove Lemma~\ref{lem:lowentropyconvolution}. Therefore, 
the constant $17/12$ in Lemma~\ref{lem:hwtupperbd} is not so critical except 
that it is close enough to 1 for our purposes.
\end{rmk}

\begin{lemma} \label{lem:lowentropyconvolution}
 Let $X, Y$ be random variables taking values in $\zq$ such that $H(X) \geq 
H(Y)$, and assume $0 < \epsilon, \epsilon' \leq \min\{\epsilon_1,\epsilon_2\}$, 
where $\epsilon_1 = \frac{1}{500}$ and $\epsilon_2 = \frac{1}{(q-1)^4}$. 
Suppose that $X$ 
has mass $1-\epsilon$ on one symbol, while $Y$ has mass $1-\epsilon'$ on a 
symbol. Then,
\begin{eqnarray}
 H(X+Y) - \frac{2H(X)+H(Y)}{3} \geq \frac{1}{51}\cdot H(Y)(1-H(Y)). 
\label{eq:lentinc}
\end{eqnarray}
\end{lemma}

\noindent\textbf{Overview of proof}. The idea is that $\epsilon, \epsilon'$ are 
small enough that we are able to invoke Lemmas~\ref{lem:hwtlowbd} and 
\ref{lem:hwtupperbd}. In particular, we show that $X+Y$ also has high weight 
on a particular symbol, which allows us to use Lemma~\ref{lem:hwtlowbd} to 
bound $H(X+Y)$ from below. Furthermore, we use Lemma~\ref{lem:hwtupperbd} in 
order to bound $H(X)$, $H(Y)$, and, therefore, $\frac{2H(X)+H(Y)}{3}$ from 
above. This gives us the necessary entropy increase for the left-hand side of 
\ref{eq:lentinc}. Note that the constant $1/51$ on the right-hand side of 
\ref{eq:lentinc} is not of any particular importance, and we have not made any 
attempt to optimize the constant.

\begin{proof}
 Let $j\in \zq$ such that $\Pr[X=j] = 1-\epsilon$, and let $j'\in \zq$ 
such 
that $\Pr[X=j'] = 1-\epsilon'$. Then,
\begin{eqnarray}
\Pr[X+Y = j+j'] \geq (1-\epsilon)(1-\epsilon')
\geq \left(\frac{499}{500}\right)^2.  \label{eq:wtprod}
\end{eqnarray}
(In a slight abuse of notation, $j+j'$ will mean $j+j'\ \text{(mod $q$)}$.)

Similarly, let us find an upper bound on $\Pr[X+Y = j+j']$. Let $p$ and $p'$ be
the underlying probability distributions of $X$ and $X'$, respectively. Then,
observe that $\Pr[X+Y=j+j']$ can be bounded 
from above as follows:
\begin{eqnarray}
 \sum_{k=0}^{q-1} p(k)p'(j+j'-k) &=& p(j)p'(j') + \sum_{k\neq j} 
p(k)p'(j+j'-k)\nonumber\\
&\leq& (1-\epsilon)(1-\epsilon') + \sum_{k\neq j} \left(\frac{p(k) + 
p'(j+j'-k)}{2}\right)^2\nonumber\\
&\leq& (1-\epsilon)(1-\epsilon') + \left(\frac{\sum_{k\neq j} (p(k) + 
p'(j+j'-k))}{2}\right)^2\nonumber\\
&=& (1-\epsilon)(1-\epsilon') + \left(\frac{\sum_{k\neq j} p(k) + \sum_{k\neq 
j'} p'(k)}{2}\right)^2\nonumber\\
&=& (1-\epsilon)(1-\epsilon') + \left(\frac{\epsilon +
\epsilon'}{2}\right)^2\nonumber\\
&=& 1 - \left(\epsilon + \epsilon' - \frac{3}{2}\epsilon\epsilon' - 
\frac{\epsilon^2}{4} - \frac{\epsilon'^2}{4}\right)\nonumber\\
&\leq& 1 - \frac{17}{18}(\epsilon + \epsilon'). \label{eq:wtineqbd}
\end{eqnarray}
Now, by Lemma~\ref{lem:hwtupperbd}, we have
\[
 H(X) \leq \frac{17\epsilon\lg(1/\epsilon)}{12\lg q}
\]
and
\[
 H(Y) \leq \frac{17\epsilon'\lg(1/\epsilon')}{12\lg q}.
\]
Also, by (\ref{eq:wtprod}) and (\ref{eq:wtineqbd}), we know that $X$ has mass
$1-\delta$ on a symbol, where $\frac{17}{18}(\epsilon + \epsilon') \leq \delta <
\frac{1}{e}$. Thus, by Lemma~\ref{lem:hwtlowbd} and Fact~\ref{fact:incrfunc}, we
have
\begin{eqnarray}
 H(X+Y) - \frac{2H(X) + H(Y)}{3} &\geq&	H(X+Y) - \frac{17}{18\lg 
q}\epsilon\lg(1/\epsilon) - \frac{17}{36\lg q}\epsilon'\lg(1/\epsilon') 
\nonumber\\
&\geq& \frac{1}{\lg q}\left(\frac{17}{18}(\epsilon +
\epsilon')\lg\left(\frac{1}{\frac{17}{18}(\epsilon + 
\epsilon')}\right) - \frac{17}{18}\epsilon\lg(1/\epsilon) - 
\frac{17}{36}\epsilon'\lg(1/\epsilon')\right) \nonumber\\
&\geq& \frac{1}{\lg 
q}\left(\frac{17}{18}(17\epsilon'+\epsilon')\lg\left(\frac{1}{\frac{17}{18}
(17\epsilon'+\epsilon')}\right)  
\right. \nonumber \\
& & \left. - \frac{17}{18}(17\epsilon') \lg(1/17\epsilon') 
\vphantom{\left(\frac{1}{\frac{17}{18}
(17\epsilon'+\epsilon')}\right)} - \frac{17}{36}\epsilon'\lg(1/\epsilon') 
\right) \label{eq:derstuff} \\
&=& \frac{1}{\lg q} \left(\frac{17}{18}\epsilon'\lg(1/17\epsilon') - 
\frac{17}{36}\epsilon'\lg(1/\epsilon')\right) \nonumber \\
&\geq& \frac{1}{36\lg q}\epsilon'\lg(1/\epsilon') \nonumber\\
&\geq& \frac{1}{51} H(Y)(1-H(Y)), \nonumber
\end{eqnarray}
were (\ref{eq:derstuff}) follows from the fact that
\begin{eqnarray*}
\frac{d}{d\epsilon}\left(\frac{17}{18}(\epsilon+\epsilon')\lg\left(\frac{1
}{\frac{17}{18} (\epsilon+\epsilon')}\right) - 
\frac{17}{18}\epsilon\lg(1/\epsilon) - 
\frac{17}{36}\epsilon'\lg(1/\epsilon')\right) = 
\frac{17}{18}\left(\lg\left(\frac{\epsilon}{\frac{17}{18}(\epsilon+\epsilon')} 
\right)\right),
\end{eqnarray*}
which is negative for $\epsilon < 17\epsilon'$ and positive for $\epsilon > 
17\epsilon'$.
\end{proof}

\subsection{High Entropy Variables}
For the remainder of this section, let $f(x) = -\frac{x\lg x}{\lg q}$. The 
following lemma proves lower and upper bounds on $f(x)$.
\begin{lemma}\label{lem:taylorapprox}
For $-\frac{1}{q} \leq t \leq \frac{q-1}{q}$, we have
\begin{eqnarray}
 \frac{1}{q} + \left(1-\frac{1}{\ln q}\right)t - \frac{q}{\ln q}t^2 \leq 
f\left(\frac{1}{q}+t\right) \leq \frac{1}{q} + \left(1 - \frac{1}{\ln 
q}\right)t - \frac{q(q\ln q - (q-1))}{(q-1)^2 \ln q}t^2. \label{eq:taylorbd}
\end{eqnarray}
\end{lemma}
\begin{proof}
 Let
 \[
  g(t) = f\left(\frac{1}{q}+t\right) - \frac{1}{q} - \left(1-\frac{1}{\ln 
q}\right)t + \frac{q}{\ln q}t^2.
 \]
To prove the lower bound in (\ref{eq:taylorbd}), it suffices to show that $g(t) 
\geq 0$ for all $-\frac{1}{q}\leq t\leq\frac{q-1}{q}$. Note that the first and 
second derivatives of $g$ are
\begin{eqnarray*}
 g'(t) &=& -\frac{\ln\left(\frac{1}{q}+t\right)}{\ln q} - 1 + \frac{2qt}{\ln 
q}\\
g''(t) &=& -\frac{1}{\left(\frac{1}{q}+t\right)\ln q} + \frac{2q}{\ln q}.
\end{eqnarray*}
It is clear that $g''(t)$ is an increasing function of 
$t\in\left(-\frac{1}{q}, \frac{q-1}{q}\right)$, and $g''(-1/2q) = 0$. Since 
$g'(-1/2q) = \frac{\ln 2 - 1}{\ln q} < 0$, it follows that $g(t)$ is minimized 
either at $t=-1/q$ or at the unique value of $t > -\frac{1}{2q}$ for which
$g'(t) 
= 0$. Note that this latter value of $t$ is $t=0$, at which $g(t) = 0$. 
Moreover, $g(-1/q) = 0$. Thus, $g(t) \geq 0$ on the desired domain, which 
establishes the lower bound.

Now, let us prove the upper bound in (\ref{eq:taylorbd}). Define
\[
 h(t) = \frac{1}{q} + \left(1-\frac{1}{\ln q}\right) t - \frac{q(q\ln q 
-(q-1))}{(q-1)^2 \ln q}t^2 - f\left(\frac{1}{q}+t\right).
\]
Note that it suffices to show that $h(t) \geq 0 $ for all $-\frac{1}{q} \leq t 
\leq \frac{q-1}{q}$. Observe that the first and second derivatives of $h$ are
\begin{eqnarray*}
 h'(t) &=& 1 - \frac{2q(q\ln q - (q-1))}{(q-1)^2\ln q}t + 
\frac{\ln\left(\frac{1}{q}+t\right)}{\ln q}\\
h''(t) &=& -\frac{2q(q\ln q - (q-1))}{(q-1)^2 \ln q} + 
\frac{1}{\left(\frac{1}{q}+t\right)\ln q}.
\end{eqnarray*}
Now, observe that $h'(0) = 0$ and $h''(0) > 0$. Moreover, $h''(t)$ is 
decreasing on $t\in \left(-\frac{1}{q},\frac{q-1}{q}\right)$. Thus, it follows 
that the minimum value of $h(t)$ occurs at either $t=0$ or $t=\frac{q-1}{q}$.
Since $h(0) = 
h\left(\frac{q-1}{q}\right) = 0$, we must have that $h(t)\geq 0$ on the desired
domain, which 
establishes the upper bound.
\end{proof}

Next, we prove a lemma that provides lower and upper bounds on the entropy of a 
distribution that samples each symbol in $\zq$ with probability close to 
$\frac{1}{q}$.
\begin{lemma} \label{lem:hdev}
 Suppose $p$ is a distribution on $\zq$ such that
for each $0\leq i\leq q-1$, we have $p(i) = \frac{1}{q} + \delta_i$ with
$\max_{0\leq i < q} |\delta_i| = \delta$. Then,
\[
 1-\frac{q^2}{\ln q}\delta^2 \leq H(p) \leq 1-\frac{q^2 (q\ln q - 
(q-1))}{(q-1)^3 \ln q}\delta^2.
\]
\end{lemma}
\begin{proof}
 Observe that $\sum_{i=0}^{q-1} \delta_i = 0$. Thus, for the lower bound on
$H(p)$, note that
\begin{eqnarray*}
 H(p) &=& \sum_{i=0}^{q-1} f\left(\frac{1}{q}+\delta_i\right)\\
 &\geq& \sum_{i=0}^{q-1} \left(\frac{1}{q} + \left(1-\frac{1}{\ln
q}\right)\delta_i - \frac{q}{\ln q}\delta_i^2\right)\\
&=& 1 - \frac{q}{\ln q}\sum_{i=0}^{q-1}\delta_i^2\\
&\geq& 1 - \frac{q^2}{\ln q}\delta^2,
\end{eqnarray*}
where the second line is obtained using Lemma~\ref{lem:taylorapprox}, and the 
final
line uses the fact that $|\delta_i|\leq \delta$ for all $i$.

Similarly, note that the upper bound on $H(p)$ can be obtained as follows:
\begin{eqnarray*}
 H(p) &=& \sum_{i=0}^{q-1} f\left(\frac{1}{q}+\delta_i\right)\\
&\leq& \sum_{i=0}^{q-1} \left(\frac{1}{q} + \left(1-\frac{1}{\ln
q}\right)\delta_i - \frac{q(q\ln q - (q-1))}{(q-1)^2\ln q}\delta_i^2\right)\\
&=& 1 - \frac{q(q\ln q - (q-1))}{(q-1)^2\ln q} \sum_{i=0}^{q-1} \delta_i^2\\
&\leq& 1 - \frac{q^2(q\ln q - (q-1))}{(q-1)^3\ln q}\delta^2,
\end{eqnarray*}
where we have used the fact that
\[
 \sum_{i=0}^{q-1} \delta_i^2 \geq \delta^2 + 
(q-1)\cdot\left(\frac{\delta}{q-1}\right)^2 = \frac{q}{q-1}\delta^2.
\]
\end{proof}

\begin{rmk}
 Lemma~\ref{lem:hdev} shows that if $p$ is a distribution over $\zq$ with 
$\max_{0\leq i < q} | p(i) - \frac{1}{q} | = \delta$, then $H(p) = 
1 - \Theta_q(\delta^2)$.
\end{rmk}

\begin{lemma} \label{lem:uniformconv}
  Let $X$ and $Y$ be random variables taking values in $\zq$ such that 
$H(X) 
\geq H(Y)$. Also, assume $0 < \delta, \delta' \leq \frac{1}{2q^2}$. Suppose 
$\Pr[X=i] = \frac{1}{q} + \delta_i$ and $\Pr[Y=i] = 
\frac{1}{q} + \delta'_i$ for $0\leq i\leq q-1$, such that $\max_{0\leq i < q} 
|\delta_i| = \delta$ and $\max_{0\leq i < q} |\delta'_i| = \delta'$. Then,
\begin{eqnarray}
 H(X+Y) - H(X) \geq \frac{\ln q}{16q^2}\cdot H(X)(1-H(X)). \label{eq:hientgain}
\end{eqnarray}
\end{lemma}

\noindent \textbf{Overview of proof}. We show that since $X$ and $Y$ 
sample all symbols in $\zq$ with probability close to $1/q$, it follows that 
$X+Y$ also samples each symbol with probability close to $1/q$. In particular, 
one can show that $X+Y$ samples each symbol with probability in 
$\left[\frac{1}{q} - \frac{\delta}{2q}, \frac{1}{q} + 
\frac{\delta}{2q}\right]$. Thus, we can use Lemma~\ref{lem:hdev} to get a lower 
bound on $H(X+Y)$. Similarly, Lemma~\ref{lem:hdev} also gives us an upper bound 
on $H(X)$. This allows us to bound the left-hand side of (\ref{eq:hientgain}) 
adequately.

\begin{proof}
By Lemma~\ref{lem:hdev}, we know that
\begin{eqnarray}
 1 - \frac{q^2}{\ln q}\delta^2 \leq  H(X) \leq 1 - \frac{q^2(q\ln q 
- (q-1))}{(q-1)^3 \ln q }\delta^2. \label{eq:xentropy}
\end{eqnarray}

Note that
\begin{align*}
 \Pr[X+Y=k] &= \sum_{i=0}^{q-1} \Pr[X=i]\Pr[Y=k-i]\\
 &= \sum_{i=0}^{q-1} \left(\frac{1}{q} + 
\delta_i\right)\left(\frac{1}{q}+\delta'_{k-i}\right)\\
&= \frac{1}{q} + \sum_{i=0}^{q-1} \delta_i \delta'_{k-i} \\
& \leq \frac{1}{q} + q\delta\delta'\\
&\leq \frac{1}{q} + \frac{\delta}{2q}.
\end{align*}
Similarly,
\begin{equation*}
 \Pr[X+Y=k] = \frac{1}{q} + \sum_{i=0}^{q-1} \delta_i\delta_{k-i} \geq \frac{1}{q} - q\delta\delta' \geq \frac{1}{q} - \frac{\delta}{2q}
\ . 
\end{equation*}
Thus, Lemma~\ref{lem:hdev} implies that
\begin{equation}
 H(X+Y) \geq 1 - \frac{q^2}{\ln q}\left( \frac{\delta}{2q} 
\right)^2 
= 1 - \frac{1}{4\ln q}\delta^2. \label{eq:sumlbd}
\end{equation}
Therefore, by (\ref{eq:xentropy}) and (\ref{eq:sumlbd}), we have
\begin{eqnarray*}
 H(X+Y) - H(X) &\geq& \left(1 - \frac{1}{4\ln q}\delta^2\right) - 
\left(1 - \frac{q^2(q\ln q - (q-1))}{(q-1)^3 \ln q }\delta^2\right)\\
&=& \left(\frac{q\ln q - (q-1)}{(q-1)^3} - 
\frac{1}{4q^2}\right)\cdot\frac{q^2}{\ln q}\delta^2\\
&\geq& \frac{\ln q}{16q^2}\cdot\frac{q^2}{\ln q}\delta^2\\
&\geq& \frac{\ln q}{16q^2}(1-H(X))\\
&\geq& \frac{\ln q}{16q^2}H(X)(1-H(X)),
\end{eqnarray*}
as desired.
\end{proof}

\section{Rough Polarization} \label{sec:polarizeproofs}
\begin{proof}[Proof of Lemma~\ref{lem:sqrtdecay}:]
 Fix a $0\leq j < 2^n$. Also, let $h = H(W_n^{(j)})$, and let $\delta = 
H((W_n^{(j)})^-) - H(W_n^{(j)}) = 
H(W_n^{(j)}) - H((W_n^{(j)})^+)$. Then, note that
\begin{eqnarray}
 \sqrt{T(W_{n+1}^{(2j)})} + \sqrt{T(W_{n+1}^{(2j+1)})} 
= \sqrt{h(1-h)+(1-2h)\delta-\delta^2}+\sqrt{ 
h(1-h)-(1-2h)\delta-\delta^2}. \label{eq:tsum}
\end{eqnarray}
For ease of notation, let $f: [-1,1]\to\R$ be the function given by
\[
 f(x) = \sqrt{h(1-h) + x} + \sqrt{h(1-h)-x}.
\]
By symmetry, we may assume that $h\leq\frac{1}{2}$ without loss of
generality. Moreover, if we let $\alpha=\alpha(q)$ be the constant described
in Theorem~\ref{thm:main-intro}, then we know that $\delta \geq \alpha h(1-h)$.
Then, since $f'''(x) \leq 0$ for $0\leq x\leq h(1-h)$, Taylor's Theorem implies
that
\begin{eqnarray*}
 \sqrt{T(W_{n+1}^{(2j)})} + \sqrt{T(W_{n+1}^{(2j+1)})} &\leq& f((1-2h)\delta)\\
&\leq& f(0) + f'(0) ((1-2h)\delta) + \frac{f''(0)}{2} ((1-2h)\delta)^2\\
&=& 2\sqrt{h(1-h)} - \frac{((1-2h)\delta)^2}{4(h(1-h))^{3/2}}\\
&\leq& 2\sqrt{h(1-h)} - \frac{(\alpha h(1-h)(1-2h))^2}{4(h(1-h))^{3/2}}\\
&=& 2\sqrt{h(1-h)} - \frac{\alpha^2}{4}(1-2h)^2 \sqrt{h(1-h)}.
\end{eqnarray*}
Thus, if $1-2h \geq \frac{\alpha}{8+\alpha}$, then the desired result follows
for $\Lambda \geq 1 - \frac{1}{2}\left(\frac{\alpha^2}{16+2\alpha}\right)^2$.

Next, consider the case in which $1-2h < \frac{\alpha}{8+\alpha}$. Then,
$\frac{4}{8+\alpha} < h \leq \frac{1}{2}$. Hence, $\delta \geq \alpha h(1-h)
\geq \frac{2\alpha}{8+\alpha}$, which implies that $\delta \geq 2(1-2h)$. It 
follows that
\[
 (1-2h)\delta-\delta^2 \leq -\frac{\delta^2}{2}.
\]
Hence, by plugging this into (\ref{eq:tsum}), we have that
\begin{eqnarray*}
 \frac{1}{2}\left(\sqrt{T(W_{n+1}^{(2j)})} + \sqrt{T(W_{n+1}^{(2j+1)})}\right) 
\leq \sqrt{h(1-h)-\frac{\delta^2}{2}}
\end{eqnarray*}
Now, recall that $\delta\geq \frac{2\alpha}{8+\alpha}$, a constant 
bounded away from 0. Moreover, if $c$ is a positive constant, then 
$\frac{\sqrt{x-c}}{\sqrt{x}}$ is an increasing function of $x$ for $x>c$. 
Since $h(1-h) \leq \frac{1}{4}$, it follows that
\[
\dfrac{\frac{1}{2}\left(\sqrt{T(W_{n+1}^{(2j)})} +  
\sqrt{T(W_{n+1}^{(2j+1)})}\right)}{T(W_n^{(j)})} \leq 
\frac{\sqrt{h(1-h)-\frac{\delta^2}{2}}}{\sqrt{h(1-h)}} \leq 
\frac{\sqrt{\frac{1}{4}-\frac{\delta^2}{2}}}{\sqrt{\frac{1}{4}}} \leq 
\sqrt{1-\frac{8\alpha^2}{(8+\alpha)^2}}.
\]

We conclude that the desired statement holds for $\Lambda = 
\max\left\{1-\frac{1}{2}\left(\frac{\alpha^2}{16+2\alpha}\right)^2, 
\sqrt{1-\frac{8\alpha^2}{(8+\alpha)^2}}\right\}$.
\end{proof}

\section{Fine polarization: Proof of Theorem~\ref{thm:finepolarize}}
\label{app:fine-polar}

\finepolarizethm*
\begin{proof}
 Let $\rho \in (\Lambda^2, 1)$ be a fixed constant, where $\Lambda$ is the 
constant described in Theorem~\ref{thm:roughpolarize}, and choose $\gamma 
> \lg(1/\rho)$ such that $\beta = \left(1+\frac{1}{\gamma}\right)\delta <
\frac{1}{2}$. Then, let us set $m = \left\lfloor \frac{n_0}{1+\gamma} 
\right\rfloor$ and $n = \left\lceil \frac{\gamma n_0}{1+\gamma} \right\rceil$,
so 
that $n_0 = m + n$. Moreover, let $d = \left\lfloor
\frac{12n\lg q}{m\lg(1/\rho)}\right\rfloor$ and choose a constant $a_\rho > 0$
such that
\[
 a_\rho > \frac{12(\ln 2)(\lg q)}{(1-2\beta)^2 \lg(1/\rho)}\left(1 + 
\lg\left(\frac{48\gamma\lg q}{\lg(1/\rho)}\right)\right).
\]

Now, we choose
\begin{eqnarray}
n_0 > (1+\gamma)\max\left\{2b_\rho \lg(2/\epsilon), 
\frac{24\lg(1/\beta)\lg q}{\beta\lg(1/\rho)}, 2a_\rho \lg(2/\epsilon), 1, 
\frac{1}{\gamma}\right\}, \label{eq:n0bd}
\end{eqnarray}
where $b_\rho$ is the constant described in Theorem~\ref{thm:roughpolarize}. 
Note that this guarantees that
\begin{eqnarray}
 m > \max\left\{b_\rho \lg(2/\epsilon), \frac{12 \lg(1/\beta)\lg q 
}{\beta\lg(1/\rho)}, a_\rho \lg(2/\epsilon) \right\}. \label{eq:mbd}
\end{eqnarray}

Then, Theorem~\ref{thm:roughpolarize} implies that there exists a set 
\begin{eqnarray}
\mathcal{W}' \subseteq \{W_m^{(i)}: 0\leq i \leq 2^m-1\}
\end{eqnarray}
such that for all 
$M\in\mathcal{W}'$, we have $\zmax(M)\leq 2\rho^m$ and 
\begin{eqnarray}
\Pr_i[W_m^{(i)}\in\mathcal{W}'] \geq 1-H(W)-\frac{\epsilon}{2}. 
\label{eq:wmiprob}
\end{eqnarray}
Let $T$ be 
the set of indices $i$ for which $W_m^{(i)} \in \mathcal{W}'$.

Fix an arbitrary $M\in\mathcal{W}'$. Recursively define 
$\left\{\tilde{Z}_k^{(i)}\right\}_{0\leq i\leq 2^k-1}$ by $\tilde{Z}_0^{(0)} = 
\zmax(M)$ and 
\begin{eqnarray*}
 \tilde{Z}_{k+1}^{(i)} = \begin{cases}
                          \left(\tilde{Z}_k^{\lfloor i/2\rfloor}\right)^2, 
\quad &i\equiv 1\pmod{2}\\
q^3 \tilde{Z}_k^{\lfloor i/2\rfloor}, \quad &i\equiv 0\pmod{2}
                         \end{cases}.
\end{eqnarray*}
Now, let us define the sets $G_j(n)\subseteq \{i\in\Z: 0\leq i \leq 2^{n}-1\}$,
for $j=0,1,\dots,d-1$ as follows:
\[
 G_j(n) = \left\{i: \sum_{\frac{jn}{d}\leq k < \frac{(j+1)n}{d}} i_k \geq \beta 
n/d\right\},
\]
where $\overline{i_{n-1}i_{n-2}\cdots i_0}$ is the binary representation of 
$i$. Also, let $G(n) = \bigcap_{0\leq j < d} G_j(n)$. Note that if we choose 
$i$ uniformly among $0,1,\dots, 2^n-1$, then $i_0, i_1, \dots, i_{n-1}$ are 
i.i.d. Bernoulli random variables. Thus, Hoeffding's inequality implies that
\begin{eqnarray*}
 \Pr_{0\leq i < 2^n} [i\in G_j(n)] \geq 1-\exp(-(1-2\beta)^2 n/2d)
\end{eqnarray*}
for every $j$. Hence, by the union bound,
\begin{eqnarray}
 \Pr_{0\leq i < 2^n} [i\in G(n)] \geq 1-d\exp(-(1-2\beta)^2 n/2d). 
\label{eq:gnunion}
\end{eqnarray}

Now, assume $i\in G(n)$. Note that $\tilde{Z}_{(j+1)n/d}^{\left(\lfloor 
i/2^{n(d-j-1)/d} 
\rfloor\right)}$ 
can be obtained by taking $\tilde{Z}_{jn/d}^{\left(\lfloor i/2^{n(d-j)/d} 
\rfloor\right)}$ and performing a sequence of $n/d$ operations, each of which 
is either $z\mapsto z^2$ (squaring) or $z\mapsto q^3 z$ 
($q^3$-fold increase). Since 
$i\in G_j(n)$, at least $\beta n/d$ of the operations must be squarings. 
Hence, it is not too difficult to see that the maximum possible value of 
$\tilde{Z}_{(j+1)n/d}^{\left(\lfloor i/2^{n(d-j-1)/d}\rfloor\right)}$ is 
obtained when we have $(1-\beta)n/d$ $q^3$-fold increases followed by 
$\beta n/d$ squarings. Hence,
\[
 \lg \tilde{Z}_{(j+1)n/d}^{\left(\lfloor i/2^{n(d-j-1)/d} \rfloor\right)} \leq 
2^{\beta n/d} \left(\frac{n}{d} (1-\beta)(3 \lg q) + \lg
\tilde{Z}_{jn/d}^{\left(\lfloor 
i/2^{n(d-j)/d}\rfloor\right)}\right).
\]
Making repeated use of the above inequality, we see that
\begin{eqnarray}
 \lg Z(M_n^{(i)}) &\leq& \lg \tilde{Z}_n^{(i)} \nonumber\\
 &\leq& 2^{\beta n} \lg \zmax(M) + 
\frac{n}{d}(1-\beta)(3\lg q)\left(2^{\beta n/d} + 2^{2\beta n/d} + \cdots + 
2^{\beta n}\right) \nonumber\\
&\leq& 2^{\beta n} \lg \zmax(M) + \frac{n}{d}(3\lg q)\frac{(1-\beta) 2^{\beta 
n}}{1-2^{-\frac{\beta n}{d}}}\nonumber\\
&\leq& 2^{\beta n}\left(\lg(2\rho^m) + \frac{n}{d}(3\lg 
q)\right)\label{eq:onefrac}\\
&\leq& -2^{\beta n}, \label{eq:minusone}
\end{eqnarray}
where (\ref{eq:onefrac}) follows from (\ref{eq:mbd}) and
\begin{eqnarray*}
 2^{-\frac{n}{d}\beta} &\leq& 2^{-\frac{\beta m\lg(1/\rho)}{12\lg q}}\\
 &\leq& \beta,
\end{eqnarray*}
while (\ref{eq:minusone}) follows from (\ref{eq:mbd}) and
\begin{eqnarray*}
 \lg(2\rho^m) + \frac{n}{d}(3\lg q) &\leq& \lg(2\rho^m) + \dfrac{3n\lg 
q}{\frac{6n\lg q}{m\lg(1/\rho)}}\\
 &\leq& 1 - m\lg(1/\rho) + \frac{m\lg(1/\rho)}{2}\\
 &=& 1 - \frac{m\lg(1/\rho)}{2}\\
 &\leq& -1.
\end{eqnarray*}
Therefore, for any $0\leq k < 2^{n_0}$ that can be written as $k = 2^n i' + 
i$, for $0\leq i' < 2^m$ and $0\leq i < 2^n$ such that $i'\in T$ and $i\in 
G(n)$, we have that for $M = W_m^{(i')}$,
\[
 \lg \zmax(W_{n_0}^{(k)}) = \lg \zmax(M_n^{(i)}) \leq -2^{\beta n} \leq
-2^{\delta n_0}.
\]
Moreover, by (\ref{eq:wmiprob}), (\ref{eq:gnunion}), and the union bound, we 
see that the probability that a uniformly chosen $0\leq k < 2^{n_0}$ is of the 
above form is at least
\begin{eqnarray*}
 1 - H(W) - \frac{\epsilon}{2} - d e^{-\frac{(1-2\beta)^2 n}{2d}} 
&\geq& 1 - H(W) - \frac{\epsilon}{2} - \frac{12n\lg 
q}{m\lg(1/\rho)}\exp\left(-\frac{(1-2\beta)^2 m\lg(1/\rho)}{12\lg q} 
\right)\\
&\geq& 1 - H(W) - \frac{\epsilon}{2} - \frac{48\gamma\lg 
q}{\lg(1/\rho)}\exp\left(-\frac{(1-2\beta)^2 m\lg(1/\rho)}{12\lg q} 
\right)\\
&\geq& 1-H(W)-\frac{\epsilon}{2} - \frac{48\gamma\lg 
q}{\lg(1/\rho)}\left(\frac{\epsilon}{2}\right)^{\frac{a_{\rho} (1-2\beta)^2 
\lg(1/\rho)}{12(\ln 2)(\lg q)}}\\
&\geq& 1-H(W)-\frac{\epsilon}{2}-\frac{48\gamma\lg 
q}{\lg(1/\rho)}\left(\frac{\epsilon}{2}\right)^{1+\lg\left(\frac{48\gamma\lg 
q}{\lg(1/\rho)}\right)}\\
&\geq& 1-H(W)-\frac{\epsilon}{2}-\frac{48\gamma\lg 
q}{\lg(1/\rho)}\cdot\frac{\epsilon}{2}\cdot\left(\frac{1}{2}\right)^{
\lg\left(\frac{48\gamma\lg q}{\lg(1/\rho)}\right)}\\
&=& 1-H(W)-\epsilon.
\end{eqnarray*}
So if we take $c_\delta = \max\left\{4(1+\gamma)a_\rho, 4(1+\gamma)b_\rho,
1+\gamma, \frac{1+\gamma}{\gamma}, \frac{24(1+\gamma)\lg(1/\beta)\lg
q}{\beta\lg(1/\rho)}\right\}$, then $n_0 > c_\delta\lg(1/\epsilon)$ would
guarantee (\ref{eq:n0bd}). This completes the proof.
\end{proof}
\end{document}